\documentclass[11pt]{article}

\usepackage[ruled]{algorithm2e}
\SetAlFnt{\small}
\SetAlCapFnt{\small}
\SetAlCapNameFnt{\small}
\SetAlCapHSkip{0pt}
\IncMargin{-\parindent}

\usepackage{algorithmic}
\usepackage{nicefrac}

\usepackage{graphicx}
\usepackage{xcolor}
\usepackage[margin	=1in]{geometry}
\usepackage{amssymb}
\usepackage{amsmath}
\usepackage[square]{natbib}
\setcitestyle{numbers}
\usepackage{amsthm}
\usepackage{bbm}	

\usepackage{listings}

\theoremstyle{plain}
\newtheorem{theorem}{Theorem}[section]
\newtheorem{corollary}[theorem]{Corollary}
\newtheorem{lemma}[theorem]{Lemma}
\newtheorem{proposition}[theorem]{Proposition}
\newtheorem{claim}[theorem]{Claim}
\newtheorem{observation}[theorem]{Observation}

\theoremstyle{definition}
\newtheorem{definition}[theorem]{Definition}
\newtheorem{example}[theorem]{Example}

\newcommand{\reals}{\mathbb{R}}

\newcommand{\vals}{\textbf{v}}
\newcommand{{\val}}{v}

\newcommand{\allocs}{\textbf{S}}
\newcommand{{\alloc}}{S}
\newcommand{\alloci}[1]{S_{#1}}

\newcommand{\prices}{\textbf{p}}
\newcommand{{\price}}{p}
\newcommand{\pricei}[1]{p_{#1}}

\newcommand{\lowprices}{\mathbf{\check{p}}}
\newcommand{\highprices}{\mathbf{\hat{p}}}
\newcommand{\lowprice}{\check{p}}
\newcommand{\highprice}{\hat{p}}
\newcommand{\lowpricei}[1]{\check{p}_{#1}}
\newcommand{\highpricei}[1]{\hat{p}_{#1}}

\newcommand{\set}[1]{{\left\{#1\right\}}}
\newcommand{\red}[1]{{\textcolor{red}{#1}}}

\definecolor{awgreen}{RGB}{0, 150, 0}

\newsavebox{\savepar}

\begin{document}

\title{Two-Price Equilibrium
\thanks{This work was partially supported by the European Research Council (ERC) under the European Union's Horizon 2020 research and innovation program (grant agreement No. 866132), by the Israel Science Foundation (grant number 317/17), and by the NSF-BSF (grant number 2020788).}
}

\author{Michal Feldman\\
	Tel Aviv University\\
	\text{michal.feldman@cs.tau.ac.il}
	\and
	Galia Shabtai\\
	Tel Aviv University\\
	\text{galiashabtai@gmail.com}
	\and 
	Aner Wolfenfeld\\ 
	Tel Aviv University\\ 
	\text{anerwolf@gmail.com}
	}

\date{}

\maketitle

\begin{abstract}
Walrasian equilibrium is a prominent market equilibrium notion, but rarely exists in markets with indivisible items.
We introduce a new market equilibrium notion, called two-price equilibrium (2PE). A 2PE is a relaxation of Walrasian equilibrium, where instead of a single price per item, every item has two prices: one for the item's owner and a (possibly) higher one for all other buyers. 
Thus, a 2PE is given by a tuple $(\textbf{S},\mathbf{\hat{p}},\mathbf{\check{p}})$ of an allocation $\allocs$ and two price vectors $\mathbf{\hat{p}},\mathbf{\check{p}}$, where every buyer $i$ is maximally happy with her bundle $S_i$, given prices $\mathbf{\check{p}}$ for items in $S_i$ and prices $\mathbf{\hat{p}}$ for all other items. 
2PE generalizes previous market equilibrium notions, such as conditional equilibrium, and is related to relaxed equilibrium notions like endowment equilibrium. 
We define the {\em discrepancy} of a 2PE --- a measure of distance from Walrasian equilibrium --- as the sum of differences $\hat{p}_j-\check{p}_j$ over all items (normalized by social welfare).
We show that the social welfare degrades gracefully with the discrepancy; namely, the social welfare of a 2PE with discrepancy $d$ is at least a fraction $\frac{1}{d+1}$ of the optimal welfare.
We use this to establish welfare guarantees for markets with subadditive valuations over identical items.
In particular, we show that every such market admits a 2PE with at least $1/7$ of the optimal welfare.
This is in contrast to Walrasian equilibrium or conditional equilibrium which may not even exist.
Our techniques provide new insights regarding valuation functions over identical items, which we also use to characterize instances that admit a WE. 

\end{abstract}

\section{Introduction}

We consider a combinatorial market setting with $m$ items and $n$ buyers. 
Every buyer $i$ has a valuation function, $v_i: 2^{[m]} \rightarrow \reals^+$, which maps every subset of items to a non-negative real number.
A valuation profile is given by a vector $\vals = (v_1, \ldots, v_n)$. 
As standard, we assume that valuation functions are monotone and normalized, i.e., for every $S \subseteq T \subseteq [m], v_i(S) \le v_i(T)$ and  $v_i(\emptyset)=0$ for every $i$.

An allocation is a partition of the items among the buyers; i.e., a vector $\allocs = (\alloci{1}, \ldots, \alloci{n})$ of disjoint sets, where $\alloci{i}$ denotes the bundle allocated to buyer $i$. 
The social welfare (SW) of an allocation $\allocs$ under valuation profile $\vals$ is the sum of the buyers' valuations for their bundles, that is, $SW(\allocs,\vals) = \sum_{i \in [n]}v_i(S_i)$. 
The optimal (welfare-maximizing) allocation is denoted by $OPT(\vals)$.

Suppose every item $j$ has some price $p_j \in \reals^+$. Given a vector of prices $p_1, \ldots, p_m$, and an allocation $\allocs$, the (quasi-linear) utility of buyer $i$ is $u_i(S_i,\prices) = v_i(S_i) - \sum_{j \in S_i} p_j$.

Walrasian equilibrium (WE) is a classical and appealing market equilibrium notion that dates back to the 70's (\citet{W74}). In a WE, despite competition among buyers, every buyer is maximally happy with her bundle and the market clears. 
That is, a WE is given by a tuple ($\allocs,\prices)$ satisfying: (i) Utility maximization: $u_i(\alloci{i},\prices)\geq u_i(T,\prices)$ for every bundle $T \subseteq [m]$, and (ii) Market clearance: all items are sold.
Moreover, by the first welfare theorem (\citet{BSMJ97}), any allocation supported in a WE has optimal social welfare. 

This appealing notion, however, comes with a serious downside, namely, it rarely exists in markets. In particular, it is known to exist for a strict subclass of submodular valuations, known as \emph{gross substitutes} (\citet{KC82}), and in some precise technical sense, gross substitutes is a maximal class for WE existence (\citet{GS99}).

As a result, different relaxations of WE have been introduced and studied. A notable one is the notion of {\em conditional equilibrium} (CE) (\citet{FKL12}), which is a tuple $(\allocs,\prices)$ satisfying: (i) individual rationality: $u_i(\alloci{i},\prices)\geq 0$, (ii) outward stability:  $u_i(\alloci{i},\prices)\geq u_i(T \cup \alloci{i},\prices)$ for every bundle $T \subseteq [m]$ and (iii) market clearance (all items are sold).
That is, the difference between a WE and a CE is that it only requires that buyers do not wish to add items to their bundle, whereas a WE requires that buyers don't wish to change their bundle with any other bundle.
A CE is guaranteed to exist for every market with submodular valuations (or even a superclass of submodular, called XOS). In addition, the CE notion admits an approximate version of the first welfare theorem; namely, any allocation supported in a CE has social welfare of at least half of the optimal social welfare. However, the notion of CE has its limitations --- it may not exist even in a market with two subadditive buyers (see Example \ref{ex:no_ce_2p}).

\textbf{Two-price equilibrium}.
We introduce a new notion of equilibrium that is based on the idea that an item may be assigned more than a single price. 
Indeed, item prices often have different prices based on different buyer characteristics, such as location, time, and history.

The new notion, termed {\em two-price equilibrium} (2PE), utilizes two prices per item. A 2PE is a relaxation of Walrasian equilibrium, and generalizes other WE relaxations (e.g., conditional equilibrium).
Like WE, it is a tuple of allocation and prices that clears the market (every buyer is maximally happy and all items are sold). 
However, in contrast to WE, where every item has a single price, 2PE specifies two prices for each item: one price for the item's owner and (a possibly higher) one for all other buyers.
The utility maximization condition then states that every buyer is maximally happy with her bundle, given that she pays the low price for items in her possession, and the high price for all other items.

Formally, a 2PE is given by a tuple $(\textbf{S},\mathbf{\hat{p}},\mathbf{\check{p}})$ where $\mathbf{\hat{p}},\mathbf{\check{p}} \in \reals^{[m]}$ are the high and low prices, respectively ($\highpricei{j} \geq \lowpricei{j}$ for every item $j$), and where 
(i) Utility maximization: $v_i(S_i) - \sum_{j \in S_i }\check{p}_j \geq v_i(T) - \sum_{j \in T \cap S_i}\check{p}_j - \sum_{j \in T \setminus S_i}\hat{p}_j$ for every bundle $T \subseteq [m]$, and (ii) all items are sold.
We note that Condition (i) of 2PE can be also written as 
$v_i(S_i) - \sum_{j \in S_i \setminus T}\check{p}_j \geq v_i(T) - \sum_{j \in T \setminus S_i}\hat{p}_j$ for every bundle $T \subseteq [m]$.

A 2PE for which $\highpricei{j}=\lowpricei{j}$ for every item $j$ is a Walrasian equilibrium. 
Furthermore, one can show that $(\allocs,\prices)$ is a conditional equilibrium iff $(\allocs,\prices,0)$ is a 2PE (see Proposition \ref{prop:2pe-ce}).
The 2PE notion is related to other relaxations of WE, such as the {\em endowment equilibrium} (\cite{BDO18}, \cite{EFF19}), named after the endowment effect, discovered by Nobel laureate Richard Thaler (\cite{KKT90}, \cite{KKT91}, \cite{KTTR01}), stating that buyers tend to inflate the value of items they own. 
Moreover, as we show in Section \ref{sec:2pe_s2pa}, 2PE is also related to Nash equilibria of simultaneous item auctions --- a simple auction format that attracted much research in the last decade (\cite{BR11}, \cite{CKS16}, \cite{FFGL13}, \cite{FS20}, \cite{CKST16}, \cite{CP14}).

Clearly, a 2PE is guaranteed to exist for every market instance. Moreover, every allocation can be supported in a 2PE. Indeed, for every allocation $\allocs$, the tuple $(\allocs,\highprices,\lowprices)$ where $\highpricei{j}=\infty$ and $\lowpricei{j}=0$ for every item $j$ is a 2PE. 
Thus, arbitrarily bad allocations can be supported in a 2PE. 
This is in stark contrast to Walrasian equilibrium or conditional equilibrium, where supported allocations have optimal welfare (for WE \cite{BSMJ97}) or at least half of the optimal welfare (for CE \cite{FKL12}). 
Moreover, 2PE's in which the high and low prices of items admit a large difference seem to be far from the notion of Walrasian equilibrium.

To study 2PE's that are ``close" to WE, we 
define a new metric, called the \emph{discrepancy} of a 2PE, defined as the sum of price differences over all items, $\sum_{j \in [m]}\left(\highpricei{j}-\lowpricei{j}\right)$, normalized by the social welfare. The discrepancy of a 2PE can be viewed as a measure of the distance between a given 2PE and a Walrasian equilibrium. Indeed, a 2PE with discrepancy 0 is a WE. Thus, every 2PE with discrepancy 0 has optimal welfare. 
We then ask whether there are instances that do not admit WE, or WE relaxations (such as CE), but do admit 2PE with low discrepancy and high welfare.

A particularly interesting class of valuations is the class of {\em subadditive} valuations --- where $v(S)+v(T)\geq v(S \cup T)$ for every sets $S,T \subseteq [m]$.
This is a natural class of valuations, known to be the frontier of ``complement-free" valuations \cite{LLN06}.
Markets with subadditive valuations may not admit any WE or CE, even in cases where all the items are identical. The following question arises: 

{\bf Question:} Do markets with subadditive valuations admit 2PE's with low discrepancy and high welfare?


\subsection{Our Results}

We first show that the social welfare of a 2PE degrades gracefully with its discrepancy.
Namely, the social welfare of a 2PE with discrepancy $d$ is at least a fraction $\frac{1}{d+1}$ of the optimal social welfare.
Armed with this welfare guarantee, our goal is to show the existence of 2PE's with low discrepancy.
We establish such results for markets with subadditive valuations over identical items. 

It should be noted that the problem of efficiently allocating identical items among multiple buyers has played a starring role in classical and algorithmic mechanism design.
Identical item settings are of particular interest in our context, where a WE is guaranteed to exist for submodular valuations, but beyond submodular, even  simple instances may not admit a WE, or even a relaxed equilibrium notion, such as conditional euilibrium.

We first establish a low discrepancy result for markets with two identical subadditive valuations over identical items.

\noindent {\bf Theorem 1:} (see Theorem \ref{thm:sa_n2_q2}) Every market with 2 identical subadditive valuations over identical items admits a 2PE with discrepancy of at most 2, thus welfare of at least $1/3$ of the optimal welfare.

Moreover, we show an instance with 2 identical subadditive valuations, where the minimum discrepancy for any 2PE is $1.3895$ (see Theorem \ref{ex:d_138}).

For an arbitrary number of identical valuations over identical items we show the following: 

\noindent {\bf Theorem 2:} (see Theorem \ref{thm:sa_n_q25}) Every market with (any number of) identical subadditive valuations over identical items admits a 2PE with discrepancy of at most $2.5$, thus welfare of at least $2/7$ of the optimal welfare.

Our main result establishes a constant factor guarantee for markets with heterogeneous subadditve valuations over identical items.

\noindent {\bf Main Theorem:} (see Theorem \ref{thm:sa_n_q6}) Every market with (any number of) subadditive valuations over identical items admits a 2PE with discrepancy of at most $6$, thus welfare of at least $1/7$ of the optimal welfare.


Furthermore, we find an interesting connection between 2PE and pure Nash equilibria (PNE) of simultaneous item auctions \cite{CKS08}. In these auctions every bidder submits a bid for every item, and items are sold simultaneously, each one in a separate auction given its own bids. For example, a simultaneous second price auction (S2PA) is one where every item is sold in a 2nd price auction.

We show a correspondence between 2PEs of a market and PNE of S2PA for the corresponding market (see Proposition \ref{prop:2pe-s2pa_pne}). 
Similar correspondences have been shown for WE and PNE of simultaneous first price auctions \cite{HKMN11} and for conditional equilibria and PNE of S2PA under the no-overbidding assumption \cite{FKL12}. 

Combined with our welfare guarantees for 2PEs in markets with subadditive valuations over identical items, this correspondence implies that S2PA for such markets admit PNE (without no-overbidding) with a constant fraction of the optimal welfare. Note that S2PAs for such markets do not necessarily admit PNE with no-overbidding \cite{BR11}.



To obtain our results, we provide new tools for the analysis of valuation functions over identical items. 
Using these tools, we also establish a necessary and sufficient condition for the existence of WE given an arbitrary valuation profile over identical items (see Theorem~\ref{thm:we_over_ident}).

\paragraph{Open Problems:} 

Our model and results constitute a first step in the analysis of 2PE, and leave some open problems for future work. 
Most immediately, it would be interesting to close the gaps between the upper and lower bounds on the discrepancy of the markets we study.
In addition, it would be interesting to conduct a similar analysis for markets with {\em heterogeneous} items. 
Specifically, do markets with subadditive valuations over heterogeneous items admit a 2PE with constant discrepancy? (This is true for XOS valuations.) If the answer to this question is affirmative, then it implies that every S2PA admits a PNE with constant approximation to the optimal welfare. 
%
Finally, in Section \ref{sec:2pe_s2pa} we show that every PNE of a S2PA has a corresponding 2PE with the same allocation (see Proposition \ref{prop:2pe-s2pa_pne}). \citet{FS20} establish bounds on the price of anarchy of S2PA under a ``no underbidding" assumption for different valuation classes. It would be interesting to study whether a PNE satisfying no underbidding corresponds to a 2PE with bounded discrepancy.

\subsection{Additional Related Work} 
Our work belongs to the line of research proposing relaxed market equilibrium notions that exist quite broadly and gives good welfare guarantees. Obvious examples include the conditional equilibrium notion of \citet{FKL12} discussed above and the combinatorial Walrasian equilibrium notion introduced by \citet{FGL13}. \citet{FKL12} show that a market admits a conditional equilibrium if and only if a S2PA for the corresponding market admits a PNE with no overbidding. 
A related notion is local equilibrium, introduced by \citet{L18}, which generalizes conditional equilibrium by relaxing individual rationality and outward stability.
The endowment equilibrium notion was proposed by \citet{BDO18} to capture the endowment effect discovered by \cite{KKT90}. \citet{BDO18} showed that every market with submodular valuations admits an endowment equilibrium with at least a half of the optimal welfare. \citet{EFF19} introduced a general framework that captures a wide range of formulations for the endowment effect, and showed that stronger endowment effects can lead to existence of endowment equilibrium also in XOS markets. 
We show conditions under which one can transform an endowment equilibrium to a 2PE and vice versa. 
\citet{EFRS20} provide welfare guarantees via pricing for markets with identical items.


\section{Preliminaries}




Recall that we consider a combinatorial market setting with $n$ buyers and $m$ items, where every buyer $i$ has a valuation function $v_i:2^{[m]\rightarrow \reals^+}$ that maps every subset of items $S$ into a real number (denoted by $v_i(S)$).
In this paper we consider mainly valuations over {\em identical} items, where  $v_i:[m]\rightarrow \reals^+$, specifies the value of buyer $i$ for every number of items $k$ between 0 and $m$ (denoted by $v_i(k)$). Such valuations are also called {\em symmetric} valuations. We consider the following symmetric valuation classes\footnote{The definitions of these valuation classes for heterogeneous items appear in Appendix \ref{sec:hetro_val}.}.

\begin{itemize}
    \item Unit demand: there exist a value $a$, s.t. $v(k) = a$ for every $0 < k \le m$
	\item Additive: there exist  a value $a$, s.t. $v(k) = a \cdot k$ for every $0 < k \le m$
	\item Submodular: $v(k) - v(k-1) \ge v(k + 1) - v(k)$ for every $0 < k \le m$
	\item XOS: $v(k) \ge \frac{k}{t}\cdot v(t)$ for any $0 < k < t \le m$
	\item Subadditive: $v(k) + v(t) \ge v(k+t)$ for any $0 < k,t \le m$ s.t. $k + t \le m$
\end{itemize}

\subsection{Walrasian Equilibrium and Relaxations}

In this section we present the definitions of Walrasian equilibrium and conditional equilibrium (for general valuations). The definition of endowment equilibrium is deferred to Section \ref{sec:EE}.

\begin{definition} [\textbf{Walrasian equilibrium (WE) \cite{W74}}]
\label{def:we}
A pair  $(\mathbf{S}, \mathbf{p})$ of an allocation $\allocs = (\alloci{1}, \ldots, \alloci{n})$ and item prices $\mathbf{p} = (p_1, \ldots, p_m)$, is a Walrasian equilibrium if:
\begin{itemize}
	\item [1.] \textbf{Utility maximization}: Every buyer receives an allocation that maximizes her utility given the item prices, i.e., $v_i(S_i) - \sum_{j \in S_i}p_j \geq v_i(T) - \sum_{j \in T}p_j$ for every $i \in [n]$ and bundle $T \subseteq [m]$.
	\item [2.] \textbf{Market clearance}: All items are allocated.\footnote{More precisely, if an item $j$ is not allocated, then $\pricei{j}=0$. One can easily verify that every such unallocated item can be allocated to an arbitrary buyer, and the resulting allocation, together with the original price vector, is also a Walrasian equilibrium. For simplicity of presentation, we assume throughout the paper that all items are allocated.}
\end{itemize}
\end{definition}

\begin{definition} [\textbf{Conditional equilibrium (CE) \cite{FKL12}}]
\label{def:ce}
A pair  $(\mathbf{S}, \mathbf{p})$ of an allocation $\allocs = (\alloci{1}, \ldots, \alloci{n})$ and item prices $\mathbf{p} = (p_1, \ldots, p_m)$, is a Conditional equilibrium if:
\begin{itemize}
	\item [1.] \textbf{Individual rationality}: Every buyer has a non-negative utility, i.e., $v_i(S_i) - \sum_{j \in S_i}p_j \geq 0$ for every $i \in [n]$.
	\item [2.] \textbf{Outward Stability}: No buyer wishes to add items to her bundle, i.e., $v_i(S_i) - \sum_{j \in S_i}p_j \geq v_i(S_i\cup T) - \sum_{j \in S_i\cup T}p_j$ for every $i \in [n]$ and bundle $T \subseteq [m]$.
	\item [3.] \textbf{Market clearance}: All items are allocated.
\end{itemize}
\end{definition}

An additional interesting relaxation of WE that attracted some attention recently is the notion of {\em endowment equilibrium} \cite{BDO18,EFF19}, called after the endowment effect \cite{KKT90,KKT91,KTTR01}.
An endowment equilibrium is a Walrasian equilibrium with respect to {\em endowed valuations}, which inflate the value of items owned by the buyer. In Section~\ref{sec:EE} we discuss the relation between an endowment equilibrium and 2PE.

\section{Two-Price Equilibrium (2PE)}
\label{sec:2pe}
In this section we introduce a new equilibrium notion termed \emph{Two Price Equilibrium (2PE)}. 
As we shall see, 2PE generalizes some market equilibrium notions considered in the literature. 

A 2PE resembles a Walrasian equilibrium, but instead of one price per item, it has two prices per item: high and low. It requires that every buyer receives the bundle that maximizes her utility, given that she pays the low price on items in her bundle, and would have to pay the high price for items not in her bundle.
The formal definition follow.

\begin{definition} [\textbf{Two Price Equilibrium (2PE)}]
\label{def:2pe}
Given a valuation profile $\vals$, a triplet, $(\allocs, \highprices, \lowprices)$, of an allocation $\allocs$, and high and low price vectors $\highprices, \lowprices$, s.t. $\hat{p}_j \ge \check{p}_j \ge 0$ for every item $j \in [m]$, is called a two price equilibrium (2PE) if the following hold: 


\begin{enumerate}
	\item \textbf{Utility maximization:}
	For every bundle $T \subseteq [m]$ and every buyer $i \in [n]$: 
	$$
	v_i(S_i) - \left(\sum_{j \in S_i \cap T}\check{p}_j + \sum_{j \in S_i \setminus T}\check{p}_j \right) \geq v_i(T) - \left(\sum_{j \in S_i \cap T}\check{p}_j + \sum_{j \in T \setminus S_i}\hat{p}_j \right)
	$$
	This is equivalent to
	\begin{equation}
	\label{eqn:2pe}
	v_i(S_i) - \sum_{j \in S_i \setminus T}\check{p}_j \geq v_i(T) - \sum_{j \in T \setminus S_i}\hat{p}_j
	\end{equation}

	\item \textbf{Market clearance:} All items are allocated.
\end{enumerate}
\end{definition}

2PE generalizes both Walrasian equilibrium and conditional equilibrium.
We next present a market that admits no Walrasian equilibrium nor conditional equilibrium, and yet, the optimal allocation can be supported in a 2PE.

\begin{example}
\label{ex:no_ce_2p}
Consider a market with 2 buyers and an item set $M = \{x,y,z,w\}$. Suppose buyer 1 has the following subadditive valuations
\begin{eqnarray*}
v_1(S)
& = &
	\begin{cases}
		0 & S = \emptyset \\
		1 & 1 \leq |S| \leq 3 \\
		2 & S = M
	\end{cases}	
\end{eqnarray*}
and buyer 2 has a unit-demand valuation, where $v_2(S) = 0.9$ for every non-empty bundle. 
We claim that this market has no conditional equilibrium (CE).
To see this, consider two cases. 
Case 1: all items are allocated to buyer 1. For this allocation to be supported by a CE, $p_j \geq 0.9$ for every item $j$. However, this violates individual rationality for buyer 1. 
Case 2: buyer 2 receives a non-empty bundle. 
To satisfy individual rationality, the sum of prices in buyer 2's bundle cannot exceed $0.9$. This, however, violates outward stability for buyer 1. 
We conclude that no CE exists for this market.
The optimal allocation gives all items to buyer 1. 
One can verify that this allocation is supported by a 2PE with $\hat{p}_j = 0.9$ and $\check{p}_j = \frac{1}{3}$ for every item $j$.
Indeed, buyer 1 is maximally happy with the grand bundle, since dropping any item (or both) would give her a lower utility. Similarly, buyer 2 cannot increase her utility, since in order to obtain any item $j$, she would need to pay $\hat{p}_j=0.9$, for a utility of 0.
\end{example}



\paragraph{Relation between 2PE and other market equilibrium notions.}

Clearly, every 2PE in which $\hat{p}_j = \check{p}_j$ for every item $j$ is a WE.
That is, 
for every valuation profile $\vals$, $(\mathbf{S},\mathbf{p})$ is a WE if and only if $(\mathbf{S},\mathbf{p},\mathbf{p})$ is a 2PE for $\vals$. 




The following proposition shows that CE is a special case of 2PE as well. 
\begin{proposition}
\label{prop:2pe-ce}
For every valuation profile $\vals$, $(\mathbf{S},\mathbf{p})$ is a CE if and only if $(\mathbf{S},\mathbf{p},\mathbf{0})$ is a 2PE.
\end{proposition}

\begin{proof}
Assume that $(\mathbf{S},\mathbf{p},\mathbf{0})$ is a 2PE. We show that $(\mathbf{S},\mathbf{p})$ is a CE.
Market clearance and individual rationality holds by the definition of 2PE. For outward stability, we show that no buyer can gain from adding items to $S_i$. By utility maximization of 2PE, we have that 
for every buyer $i$ and every set $T \subseteq [m]$, $v_i(S_i) \geq v_i(T) - \sum_{j \in T \setminus S_i} p_j$. Specifically, this inequality is true for every set $T = S_i \cup T'$, where $T' \subseteq [m] \setminus S_i$. Thus, for every $T' \subseteq [m] \setminus S_i$, $v_i(S_i) \geq v_i(S_i \cup T') - \sum_{j \in T'} p_j$, or $v_i(T'|S_i) \leq \sum_{j \in T'} p_j$ which is precisely outward stability. It follows that $(\mathbf{S},\mathbf{p})$ is a CE

Now assume that $(\mathbf{S},\mathbf{p})$ is a CE. We  show that $(\mathbf{S},\mathbf{p},\mathbf{0})$ is a 2PE. Market clearance holds by the definition of CE. By outward stability, for every buyer $i \in [n]$ and every set $T \in [m] \setminus S_i$, 
$$v_i(S_i) - \sum_{j \in S_i} p_j \geq v_i(S_i \cup T) - \sum_{j \in S_i} p_j - \sum_{j \in T} p_j$$ 
i.e.,
\begin{eqnarray*}
v_i(S_i) 
& \geq & 
v_i(S_i \cup T) - \sum_{j \in T} p_j
\end{eqnarray*}
For every set $T' \subseteq [m]$, let $T = T' \setminus S_i$. Then, by monotonicity $v_i(S_i) \ge v_i(T') - \sum_{j \in T' \setminus S_i} p_j$. Thus, utility maximization follows, and $(\mathbf{S},\mathbf{p},\mathbf{0})$ is a 2PE.
\end{proof}

In Section~\ref{sec:EE} we show a strong connection between endowment equilibrium and 2PE; namely, we show how a 2PE can be transformed into an endowment equilibrium and vice versa.

\subsection{Relation Between 2PE and Simultaneous Second Price Auctions}
\label{sec:2pe_s2pa}

A simultaneous second price auction (S2PA) is a simple auction format, where, despite the complex valuations of the bidders, every bidder submits a bid on every item, and every item is sold separately in a 2nd price auction; i.e., every item is sold to the bider who submitted the highest bid for that item, and the winner pays the second highest bid for that item. 

A bid profile in a S2PE is denoted by $\mathbf{b} = (\mathbf{b}_1, \ldots, \mathbf{b}_n)$,
where $b_i = (b_{i1}, \ldots, b_{im})$ is the bid vector of bidder $i$; $b_{ij}$ being bidder $i$'s bid on item $j$, for $j=1,\ldots,m$. 
Let $S_i(\mathbf{b})$ denotes the set of items won by buyer $i$, and let $\mathbf{S}(\mathbf{b})=(S_1(\mathbf{b}),\ldots,S_n(\mathbf{b}))$ denote the obtained allocation. 
Finally, let $p_j(\mathbf{b})$ denote the price paid by the winner of item $j$ (i.e., the second highest bid on item $j$). 

A S2PA is not a truthful auction, and its performance is often measured in equilibrium. 
A bid profile is said to be a pure Nash equilibrium in a S2PA if the following holds. 


\begin{definition} 
\label{def:pne}
A bid profile 
$\mathbf{b}$ in a S2PA is a pure Nash equilibrium (PNE) if for any $i \in [n]$ and for any $b_i^{'}$,
$v_i(S_i(\mathbf{b})) - \sum_{j \in S_i(\mathbf{b})} p_j(\mathbf{b}) \ge v_i(S_i(b_i^{'},\mathbf{b}_{-i}) - \sum_{j \in S_i(b_i^{'},\mathbf{b}_{-i})} p_j(b_i^{'},\mathbf{b}_{-i})$.
\end{definition}


The following proposition shows that 
a pure Nash equilibrium of S2PA corresponds to a 2PE of the corresponding market.

\begin{proposition}
\label{prop:2pe-s2pa_pne}
Consider a valuation profile $\vals$. 
The triplet $(\mathbf{S}, \mathbf{\hat{p}}, \mathbf{\check{p}})$ is a 2PE for $\vals$ if and only if there exists a bid profile, $\mathbf{b}$, which is a PNE of the S2PA for $\mathbf{v}$ (under some tie breaking rule),
such that $\mathbf{S}(\mathbf{b}) = \mathbf{S}$, and for every item $j$, $\hat{p}_j = max_{k \in [n]} b_{kj}$ and $\check{p}_j = max2_{k \in [n]} b_{kj}$.
\end{proposition}

\begin{proof}
Assume that $\mathbf{b}$ is a S2PA PNE for $\mathbf{v}$. We show that $(\mathbf{S}(\mathbf{b}), \mathbf{\hat{p}}, \mathbf{\check{p}})$ is a 2PE. Notice that for every $j \in [m]$, $\hat{p}_j = max_{k \in [n]} b_{kj} \ge max2_{k \in [n]} b_{kj} = \check{p}_j$.
Since $\mathbf{b}$ is a PNE, for every buyer $i$, $u_i(S_i(\mathbf{b}), v_i) \ge u_i(S_i(b'_i, \mathbf{b}_{-i}), v_i)$ for every bid $b'_i$. Let $T = S_i(b'_i, \mathbf{b}_{-i})$ and substitute $\hat{p}_j = max_{k \in [n]} b_{kj}$ and $\check{p}_j = max2_{k \in [n]} b_{kj}$, then:

\begin{eqnarray*}
v_i(S_i(\mathbf{b})) - \sum_{j \in S_i(\mathbf{b})} max2_{k \in [n]} b_{kj} 
& \ge &
v_i(T) - \sum_{j \in T \cap S_i(\mathbf{b})} max2_{k \in [n]} b_{kj} - \sum_{j \in T \setminus S_i(\mathbf{b})} max_{k \in [n]} b_{kj}
\end{eqnarray*}
By subtracting $\sum_{j \in T \cap S_i(\mathbf{b})} max2_{k \in [n]} b_{kj}$ from both sides and rearranging, we get:
\begin{eqnarray*}
v_i(S_i(\mathbf{b})) - \sum_{j \in S_i(\mathbf{b}) \setminus T} max2_{k \in [n]} b_{kj} 
& \ge &
v_i(T) - \sum_{j \in T \setminus S_i(\mathbf{b})} max_{k \in [n]} b_{kj},
\end{eqnarray*}
which is precisely the utility maximization condition of 2PE.
Moreover, all items are allocated in $\mathbf{S}(\mathbf{b})$, and therefore the market clearance condition is also satisfied.
Hence, the triplet $(\mathbf{S}(\mathbf{b}), \mathbf{\hat{p}}, \mathbf{\check{p}})$ is a 2PE.

Now assume that $(\mathbf{S}, \mathbf{\hat{p}}, \mathbf{\check{p}})$ is a 2PE. We need to show that there exists a bid profile, $\mathbf{b}$, which is a S2PA PNE for $\mathbf{v}$ (under some tie breaking rule),
such that $\mathbf{S}(\mathbf{b}) = \mathbf{S}$, $\hat{p}_j = max_{k \in [n]} b_{kj}$ and $\check{p}_j = max2_{k \in [n]} b_{kj}$.
Let $\mathbf{b}$ be a bid profile, such that for every buyer $i$ and every item $j$:
$$b_{ij} = 
	\begin{cases}
		\hat{p}_j & j \in S_i \\
		\check{p}_j & j \notin S_i
	\end{cases}$$
As $\hat{p}_j \ge \check{p}$, if tie breaking is done according to $\mathbf{S} = (S_1, S_2, \ldots, S_n)$, then $\mathbf{S}(\mathbf{b}) = \mathbf{S}$. We show that $\mathbf{b}$ is a PNE. Let $b'_i$ be some bid vector of buyer $i$, and let $T = S(b'_i, \mathbf{b}_{-i})$. Then,
\begin{eqnarray*}
u_i(S_i(\mathbf{b}), v_i)
& = &
v_i(S_i(\mathbf{b})) - \sum_{j \in S_i(\mathbf{b})} max2_{k \in [n]} b_{kj}  \\
& = & 
v_i(S_i(\mathbf{b})) - \sum_{j \in S_i(\mathbf{b}) \setminus T} max2_{k \in [n]} b_{kj} - \sum_{j \in S_i(\mathbf{b}) \cap T} max2_{k \in [n]} b_{kj} \\
& = &
v_i(S_i(\mathbf{b})) - \sum_{j \in S_i(\mathbf{b}) \setminus T} \check{p}_j - \sum_{j \in S_i(\mathbf{b}) \cap T} max2_{k \in [n]} b_{kj} \\
& \ge &
v_i(T) - \sum_{j \in T \setminus S_i(\mathbf{b})} \hat{p}_j - \sum_{j \in S_i(\mathbf{b}) \cap T} max2_{k \in [n]} b_{kj} \\
& = &
v_i(T) - \sum_{j \in T \setminus S_i(\mathbf{b})} max_{k \in [n]} b_{kj} - \sum_{j \in S_i(\mathbf{b}) \cap T} max2_{k \in [n]} b_{kj} \\
& = &
u_i(T, v_i),
\end{eqnarray*}
where the inequality follows from the fact that 
 $(\mathbf{S}, \mathbf{\hat{p}}, \mathbf{\check{p}})$ is a 2PE.
\end{proof}



\section{Discrepancy Factor of 2PE}
\label{sec:disc}

The main difference between a two-price equilibrium and a Walrasian equilibrium is the use of two prices per item (high price $\highpricei{j}$ and low price $\lowpricei{j}$) rather than a single price. 
This makes the notion of 2PE similar in spirit to WE. Namely, prices are still almost anonymous (in contrast to other approaches where prices are buyer-dependent, see, e.g., \cite{HR09}), and every buyer is maximally happy with her bundle. 
The closer the two prices $\highpricei{j}$ and $\lowpricei{j}$ are together, the better the 2PE resembles a WE. 
Indeed, in the extreme case, where $\highpricei{j}=\lowpricei{j}$ for every item $j$, the two notions coincide. 


Consequently, a natural measure of distances of a 2PE from WE is the sum of price differences over all items.
We further normalize the sum of price differences by the social welfare of the allocation, so that 
the discrepancy is independent of the units used (e.g., USD vs. Euros)\footnote{Formally, suppose $(\mathbf{S}, \mathbf{\hat{p}}, \mathbf{\check{p}})$ is a 2PE with respect to valuation profile $\mathbf{v}$, and let $\mathbf{v}'$ be a valuation profile such that $v_i'(T) = c\cdot v_i(T)$ for every buyer $i$ and bundle $T$ and some constant $c \in \reals^+$. Clearly, $(\mathbf{S}, c\cdot\mathbf{\hat{p}}, c\cdot\mathbf{\check{p}})$ is a 2PE w.r.t. $\mathbf{v}'$, which has the same discrepancy as that of $(\mathbf{S}, \mathbf{\hat{p}}, \mathbf{\check{p}})$.}.


This motivate us to define the {\em discrepancy} of a 2PE as follows.


\begin{definition} [\textbf{Discrepancy}]
\label{def:qspp}
The {\em discrepancy} of a 2PE $(\mathbf{S}, \mathbf{\hat{p}}, \mathbf{\check{p}})$ under valuation profile $\mathbf{v}$ is
\begin{eqnarray}
\label{eqn:q}
D(\mathbf{S}, \mathbf{\hat{p}}, \mathbf{\check{p}}) 
& = &
\frac{\sum_{j \in [m]} (\hat{p}_j - \check{p}_j)}{SW(\mathbf{S}, \vals)}
\end{eqnarray}
\end{definition}

Low discrepancy is a desired property; a 2PE with low discrepancy is closer in spirit to WE in both fairness and simplicity. 
As we shall soon show, low discrepancy also implies high efficiency.

The 2PE notion is appealing from an existence perspective; indeed, every allocation $\allocs$ can be supported in a 2PE by setting $\hat{p}_j = \infty$, $\check{p}_j = 0$ for every item $j$.
However, from a welfare maximization perspective, no guarantee can be given. 
This is in stark contrast to WE (where, by the 1st welfare theorem, every allocation supported in a WE has optimal welfare), and to weaker equilibrium notions, such as conditional equilibrium (where every allocation supported in a CE gives at least half of the optimal welfare \cite{FKL12}).
In contrast, an allocation supported in a 2PE may have an arbitrarily low welfare. 

The following proposition shows that the social welfare of a 2PE degrades gracefully with its discrepancy.


\begin{proposition}
\label{prop:disc_sw}
(low discrepancy implies high welfare)
	Let $(\mathbf{S}, \mathbf{\hat{p}}, \mathbf{\check{p}})$ be a 2PE for valuation $\mathbf{v}$ with discrepancy $d$. Then, $SW(\mathbf{S}, \vals) \geq \frac{1}{1+d}OPT(\vals)$.
\end{proposition}

\begin{proof}
Let $\mathbf{S}^*$ be an optimal allocation. Let $L_i = \{j | j \in S_i \cap S_i^*\}$ and $L = \bigcup_{i=1}^nL_i$. Since $(\mathbf{S}, \mathbf{\hat{p}}, \mathbf{\check{p}}) $ is a 2PE we have that for every buyer $i$ and bundle $T \subseteq M$: $$\sum_{j \in T \setminus S_i}\hat{p}_j - \sum_{j \in S_i \setminus T}\check{p}_j \geq v_i(T) - v_i(S_i).$$ 
Summing the above inequality over all $i \in [n]$ and substituting $T$ with $S_i^*$ gives $$\sum_{i \in [n]}(\sum_{j \in S_i^* \setminus S_i}\hat{p}_j -\sum_{j \in S_i \setminus S_i^*}\check{p}_j) = \sum_{j \in M \setminus L} (\hat{p}_j -\check{p}_j) \geq OPT(\vals) - SW(\mathbf{S}, \vals)$$
It follows that $$\sum_{j \in M} (\hat{p}_j -\check{p}_j) = \sum_{j \in M \setminus L} (\hat{p}_j -\check{p}_j) + \sum_{j \in L} (\hat{p}_j -\check{p}_j) \geq OPT(\vals) - SW(\mathbf{S}, \vals) + \sum_{j \in L} (\hat{p}_j -\check{p}_j) \geq OPT(\vals) - SW(\mathbf{S}, \vals)$$ 
Substitutes $d = \frac{\sum_{j \in M} (\hat{p}_j - \check{p}_j)}{SW(\mathbf{S}, \vals)}$ and rearrange to get the desired bound $SW(\mathbf{S}, \vals) \geq \frac{1}{1+d}OPT(\vals)$
\end{proof}



We also define the discrepancy of a given allocation $\mathbf{S}$ as the discrepancy of the smallest-discrepancy 2PE supporting $\mathbf{S}$.
\begin{definition}
Given valuation profile $\vals$, the discrepancy of an allocation $\mathbf{S}$ is defined as $$D(\mathbf{S}) =  min_{(\mathbf{S}, \mathbf{\hat{p}}, \mathbf{\check{p}})\in 2PE} D(\mathbf{S}, \mathbf{\hat{p}}, \mathbf{\check{p}})$$
\end{definition} 

For all reasons mentioned above, it is desirable to  identify allocations with low discrepancy.

Clearly, if $\mathbf{S}$ is supported by a WE, then $D(\mathbf{S}) = 0$. Indeed, every allocation supported by a WE has optimal welfare.

It is also known that every allocation supported by a conditional equilibrium has at least a half of the optimal welfare \cite{FKL12}.
The following proposition shows that the discrepancy of every such allocation is at most 1. 
Together with 
Proposition~\ref{prop:disc_sw}, it gives an alternative proof to the welfare guarantee of a conditional equilibrium.

\begin{proposition}
\label{prop:dce}
	Let $\mathbf{S}$ be an allocation that is supported by a conditional equilibrium. Then, $D(\mathbf{S}) \leq 1$. Moreover, this is tight.
\end{proposition}

\begin{proof}
let $(\mathbf{S} ,\mathbf{p})$ be a CE then: $$D(\mathbf{S}, \mathbf{p}, \mathbf{0}) = \frac{\sum_{i \in [n]}\sum_{j \in S_i} p_j}{SW(\mathbf{S}, \vals)} \leq \frac{\sum_{i \in [n]}v_i(S_i)}{SW(\mathbf{S}, \vals)} = 1$$ where the inequality follows by individual rationality of CE.

Next, we show an example in which the above bound is tight. Let $M = \{x, y\}$, and suppose buyers 1, 2 have unit-demand valuations where $v_1(\{x\}) = v_2(\{y\}) = 2$ and $v_1(\{y\}) = v_2(\{x\}) = 1$. Consider the allocation $S_1 = \{y\}$ and $S_2 = \{x\}$. One can easily verify that $(\mathbf{S}, \mathbf{p})$ is a CE for $p_x=p_y=1$. Now observe that $(\mathbf{S}, \mathbf{p}, \mathbf{0})$ is a 2PE for $\mathbf{v}$, and $D(\mathbf{S}, \mathbf{p}, \mathbf{0}) = 1$.
\end{proof}



It is shown in \cite{CKS08} that for any XOS valuation profile, the optimal allocation is supported by a S2PA PNE with no-overbidding.
By \cite{FKL12}, every S2PA PNE with no-overbidding can be transformed into a CE that preserves the same allocation. 
It then follows by Proposition \ref{prop:dce} that the discrepancy of the optimal allocation in every XOS valuation profile is at most 1.
In Appendix \ref{sec:d_upper_m} we show that for any general valuation profile, every welfare-maximizing allocation has a discrepancy of at most $m$.

Notably, there exist markets that admit neither Walrasian equilibrium, nor conditional equilibrium, and yet, the optimal allocation is supported by a 2PE with small discrepancy. 
For example, the market in Example~\ref{ex:no_ce_2p} admits no Walrasian equilibrium, nor conditional equilibrium, and yet, the optimal allocation, $S^*$, is supported in a 2PE with discrepancy $D(\mathbf{S}^*, \mathbf{\hat{p}}, \mathbf{\check{p}}) = \frac{3.6 - \frac{4}{3}}{2} = \frac{17}{15}$.

\section{Geometric Properties of Valuations over Identical Items}
\label{sec:vals_id}
In this section we introduce new  geometric properties of valuations over identical items, which prove useful in establishing upper bounds on the discrepancy of 2PE in such markets. 
Hereafter, we refer to a valuation over identical items as a {\em symmetric} valuation.




\begin{definition} [\textbf{max-forward-slope ($\overrightarrow{\Delta}$)}]
\label{def:fdelta}
Given a symmetric valuation $v$, and some $0 \le k < m$, and  $1 \le r \le m-k$, the $(k,r)$-max-forward-slope is defined as 
\begin{eqnarray}
\label{eqn:fdelta}
\overrightarrow{\Delta}_{v}(k,r) = \max_{l = 1, 2, \ldots, r} \set{\frac{v(k+l)-v(k)}{l}}
\end{eqnarray}
We say that $l'$ realizes $\overrightarrow{\Delta}_{v}(k,r)$, if $l'$ is the minimum index s.t. $\overrightarrow{\Delta}_{v}(k,r) = \frac{v(k+l')-v(k)}{l'}$.
In addition, we use $\overrightarrow{\Delta}_{v}(k)$ to denote $\overrightarrow{\Delta}_{v}(k,m-k)$, and refer to $\overrightarrow{\Delta}_{v}(k)$ as the $k$-max-forward-slope.
\end{definition}

The sorted vector of the max-forward slopes is defined by $\overrightarrow{\Delta}^s_v(k)$. That is, $\overrightarrow{\Delta}^s_v(k) \le \overrightarrow{\Delta}^s_v(k+1)$ for every $0 \le k < m-1$.


\begin{definition} [\textbf{min-backward-slope ($\overleftarrow{\Delta}$)}]
\label{def:bdelta}
Given a symmetric valuation $v$, and some $0 \le k < m$, and  $1 \le r \le m-k$, the $(k,r)$-min-backward-slope is defined as 
\begin{eqnarray}
\label{eqn:bdelta}
\overleftarrow{\Delta}_{v}(k,r) = \min_{l = 1, 2, \ldots, r} \set{\frac{v(k)-v(k-l)}{l}}
\end{eqnarray}
We say that $l'$ realizes $\overleftarrow{\Delta}_{v}(k,r)$, if $l'$ is the minimum index s.t. $\overleftarrow{\Delta}_{v}(k,r) = \frac{v(k)-v(k-l')}{l'}$.
In addition, we use $\overleftarrow{\Delta}_{v}(k)$ to denote $\overleftarrow{\Delta}_{v}(k,k)$, and refer to $\overleftarrow{\Delta}_{v}(k)$ as the $k$-min-backward-slope.
\end{definition}



{\bf Submodular closure:} Given a symmetric valuation $v$, the minimal submodular valuation that (point-wise) upper bounds it is called the submodular closure (SM-closure) of $v$. The SM-closure of a function is known to be unique (see, e.g., \citet{EFRS20});
%
see Figure \ref{fig:sm-closure} for an illustration.

\begin{figure}[h!]
\begin{center}
\includegraphics[width=16cm,height=8cm]{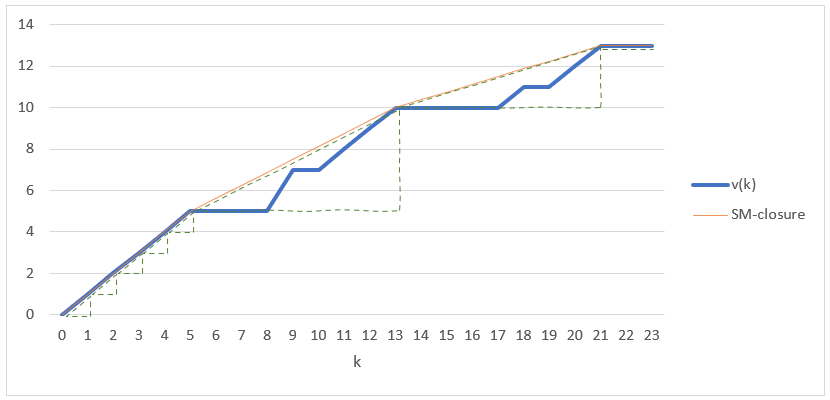}
\end{center}
\caption{The blue curve depicts a valuation $v$, and the orange curve depicts the submodular closure of $v$. The right triangles between adjacent intersecting indices are shown in green. The list of intersecting indices is: $I_v = \set{0, 1, 2, 3,4,5, 13, 21, 22, 23}$, where $i_0 = 0$, $i_1 = 1$, $\ldots$, $i_8 = 21$, $i_9 = 22$ and $i_{10} = 23$. There are $9$ right triangles, where the last two, $T_{8}$ and $T_{9}$, are degenerated, with slope $0$.}
\label{fig:sm-closure}
\end{figure}

Given a symmetric valuation function $v: [m] \rightarrow R^+$, we define the following: Let $\tilde{v}: [m] \rightarrow R^+$ be the SM-closure of $v$, and let $I_v$ be the set of indices $k \in [m]$ for which $v(k) = \tilde{v}(k)$, i.e., the set of points in which the $v$ and $\tilde{v}$ intersect. We refer to $I_v$ as the set of intersection indices.
For $0 \le l < |I_v|-1$, let $T_{{l}}$ be the right triangle between two adjacent intersecting indices, $i_{l}$ and $i_{l+1}$, with vertices $(i_{l}, v(i_{l}))$, $(i_{l+1}, v(i_{l}))$ and $(i_{l+1}, v(i_{l+1}))$ (see Figure \red{\ref{fig:sm-closure}}). Let $\alpha_l = \frac{v(i_{l+1}) - v(i_{l})}{i_{l+1} - i_{l}}$ be the \emph{slope} of the triangle $T_{{l}}$.
If $v(i_{l+1}) = v(i_{l})$, then $T_{{l}}$ is a degenerated triangle (a line), with slope $\alpha_l = 0$. 
Let $T_v = \set{T_{0}, T_{1}, \ldots, T_{{|I_v|-2}}}$ be the set of all right triangles of $v$.
For every $0 \le l < |I_v|-1$, and every $i_{l} \le k < i_{l+1}$, we say that $k \in T_{{l}}$.

In what follows we present some useful lemmas and theorems regarding symmetric valuation functions. The complete proofs, as well as additional observations, appear in Appendix \ref{sec:miss_valis_id}.







The following lemma gives a lower bound on $v(k)$ as a function of the max-forward-slopes of $v$ up to $k$. 


\begin{lemma}
\label{lem:intgr-lem}
For every symmetric subadditive valuation function $v$, and $0 < k \le m$, $v(k) \ge k \cdot \min_{0 \le k' < k} \set{\overrightarrow{\Delta}_{v}(k')}$.

\end{lemma}



The flat function of a symmetric valuation $v$ is defined as $$\phi_v(k) = 
	\begin{cases}
		0      & 0 \le k < m \\
		v(m)   & k = m
	\end{cases}$$

The following observation specifies the max-forward-slope of the flat function.


\begin{observation}
\label{obs:flat-LMFS}
Let $v$ be a symmetric function and let $\phi_v$ be its flat function.
Then, $\overrightarrow{\Delta}_{\phi_v}(k) = \frac{v(m)}{m -k}$.
\end{observation}

\begin{proof}
For $0 \le k < m$,
$$\overrightarrow{\Delta}_{\phi_v}(k) = \overrightarrow{\Delta}_{\phi_v}(k,m-k) = \max_{l = 1, 2, \ldots, r} \set{\frac{\phi_v(k+l)-\phi_v(k)}{l}} = \frac{\phi_v(m)-\phi_v(k)}{m-k} = \frac{v(m)}{m-k},$$
where the second equality follows from Definition \ref{def:fdelta} and the last two equalities follow from the definition of $\phi_v(k)$.
Hence, $\overrightarrow{\Delta}_{\phi_v}(k) = \frac{v(m)}{m-k}$
\end{proof}

We now show that given a valuation $v$ and its corresponding flat function $\phi_v$, the sorted-max-forward-slope of $v$ is at most the max-forward-slope of $\phi_v$.

\begin{theorem}
\label{thm:slv-le-lf}
For every symmetric valuation $v$, for every $0 \le k < m$, $\overrightarrow{\Delta}^s_v(k) \le \overrightarrow{\Delta}_{\phi_v}(k)$.
\end{theorem}



To prove Theorem \ref{thm:slv-le-lf}, we introduce the "reorder and unify of adjacent triangles" operation in Lemma \ref{lem:sort2tri}. The idea is to 
repeatedly switch two adjacent triangles in $v$, until the obtained valuation comprises of a single triangle.

\begin{definition}
Given a symmetric valuation $v$, a constant $c \ge 1$ and an integer $0 \le k < m$, we say that $k$ is \emph{$c-$bad} 
if $\overrightarrow{\Delta}_v(k) > c \cdot\frac{v(m)}{m}$; otherwise, we say that $k$ is \emph{$c-$good}.
\end{definition}

The following lemma establishes an upper bound on the number of {$c-$bad} numbers in $[m-1]$. 


\begin{lemma}
\label{lem:c-bad}
For every symmetric valuation $v$, for every $c \ge 1$, there are at most $m - \lfloor \frac{(c-1)}{c} \cdot m \rfloor - 1$ $c-$bad integers in $\set{0, 1, \ldots, m-1}$.
\end{lemma}

\begin{proof}
Fix $c \ge 1$ and let $\phi_v$ be $v$'s corresponding flat function. By observation \ref{obs:flat-LMFS}, for every $0 \le k < m$, $\overrightarrow{\Delta}_{\phi_v}(k) = \frac{v(m)}{m-k}$. Let $k'$ be a $c-bad$ integer w.r.t. $\phi_v$, i.e., $\overrightarrow{\Delta}_{\phi_v}(k') = \frac{v(m)}{m-k'} > c \cdot\frac{v(m)}{m}$. Rearranging we get that $k' > \frac{(c-1)}{c} \cdot m$. Since  $\overrightarrow{\Delta}_{\phi_v}(k)$ is monotonically increasing, there are exactly $m - \lfloor \frac{(c-1)}{c} \cdot m \rfloor - 1$ integers in $\phi_v$ which are $c-$bad. By Theorem \ref{thm:slv-le-lf}, for every $0 \le k < m$, $\overrightarrow{\Delta}_{v}^s(k) \le \overrightarrow{\Delta}_{\phi_v}(k)$, and therefore there are at most $m - \lfloor \frac{(c-1)}{c} \cdot m \rfloor - 1$ integers which are $c-$bad w.r.t. $v$.


\end{proof}

\section{Properties of 2PEs with Identical Items}
\label{sec:2pe_id}

In this section we present some properties of 2PEs in markets with identical items. 
We first define 2PE  with uniform prices:

\begin{definition} [\textbf{2PE with uniform prices (U-2PE)}]
\label{def:2pe_ecp}
A triplet $(\mathbf{S}, \mathbf{\hat{p}}, \mathbf{\check{p}})$ is a 2PE with uniform prices (U-2PE) for valuation profile $\vals$, if it is a 2PE for $\vals$ and for every buyer $i \in [n]$, every items $j,j' \in S_i$, $\hat{p}_j = \hat{p}_{j'}$ and $\check{p}_j = \check{p}_{j'}$. Let $\hat{p}^{(i)}$ and $\check{p}^{(i)}$ denote these prices, respectively.
\end{definition}

The following proposition shows that for studying the discrepancy in markets with identical items it is without loss of generality to restrict attention to U-2PEs.


\begin{proposition}
\label{prop:2pe_ident_same_customer_prices}
If $(\mathbf{S}, \mathbf{\hat{p}}, \mathbf{\check{p}})$ is a 2PE for some symmetric valuation profile $\vals$, then there exists a U-2PE $(\mathbf{S}, \mathbf{\hat{p}'}, \mathbf{\check{p}'})$ s.t. $D(\mathbf{S}, \mathbf{\hat{p}'}, \mathbf{\check{p}'}) = D(\mathbf{S}, \mathbf{\hat{p}}, \mathbf{\check{p}})$.
\end{proposition}

The proof of Proposition \ref{prop:2pe_ident_same_customer_prices} follows by an iterative invocation of the following lemma for every buyer $i \in [n]$. 

\begin{lemma}
\label{lem:2pe_ident_same_customer_low_prices}
Let  $(\mathbf{S}, \mathbf{\hat{p}}, \mathbf{\check{p}})$ be a 2PE for some symmetric valuation profile $\vals$. Let $l$ be some buyer. Let $\check{p}_{j}' = \frac{1}{\mid S_l \mid} \sum_{t \in S_l} \check{p}_{t}$ and $\hat{p}_{j}' = \frac{1}{\mid S_l \mid} \sum_{t \in S_l} \hat{p}_{t}$ for every item $j \in S_l$ and $\check{p}_{j}' = \check{p}_{j}$ and $\hat{p}_{j}' = \hat{p}_{j}$ for every item $j \notin S_l$. Then:
\begin{itemize}
    \item $(\mathbf{S}, \mathbf{\hat{p}}, \mathbf{\check{p}}')$ is a 2PE.
    \item $(\mathbf{S}, \mathbf{\hat{p}}', \mathbf{\check{p}})$ is a 2PE.
    \item $D(\mathbf{S}, \mathbf{\hat{p}}, \mathbf{\check{p}'}) = D(\mathbf{S}, \mathbf{\hat{p}}', \mathbf{\check{p}}) = D(\mathbf{S}, \mathbf{\hat{p}}, \mathbf{\check{p}})$
\end{itemize}
\end{lemma}

\begin{proof}
As $(\mathbf{S}, \mathbf{\hat{p}}, \mathbf{\check{p}})$ is a 2PE, $v_i(S_i) - \sum_{j \in S_i \setminus T}\check{p}_j \geq v_i(T) - \sum_{j \in T \setminus S_i}\hat{p}_j$
for every buyer $i \in [n]$ and every set $T \subseteq [m]$. 
As $\check{p}'_j = \check{p}_j$ and $\hat{p}'_j = \hat{p}_j$ for every item $j \notin S_l$, we have that, 
\begin{eqnarray}
\label{eqn:2pe_ident_1}
v_i(S_i) - \sum_{j \in S_i \setminus T}\check{p}'_j 
& \geq & 
v_i(T) - \sum_{j \in T \setminus S_i}\hat{p}_j, ~ \mbox{for every buyer $i \neq l$ and every set $T \subseteq [m]$}
\end{eqnarray}

and,
\begin{eqnarray}
\label{eqn:2pe_ident_2}
v_i(S_i) - \sum_{j \in S_i \setminus T}\check{p}_j 
& \geq & 
v_i(T) - \sum_{j \in T \setminus S_i}\hat{p}'_j, ~ \mbox{for buyer $i=l$ and every set $T \subseteq [m]$.}
\end{eqnarray}
To show that $(\mathbf{S}, \mathbf{\hat{p}}, \mathbf{\check{p}'})$ is a 2PE, 
it remains to show that (\ref{eqn:2pe_ident_1}) holds for buyer $l$. 

For a given set $T$, let $x = \mid S_l \cap T \mid$ and let $X$ be the set of the 
$x$ items with the lowest $\check{p}_j$ in $S_l$. Let $T' = X \cup (T \setminus S_l)$. Notice that $S_l \setminus T' = S_l \setminus X$ and therefore, $\sum_{j \in S_l \setminus T'} \check{p}_j \ge \sum_{j \in S_l \setminus T'} \check{p}'_j = (|S_l| - x) \cdot \frac{1}{\mid S_l \mid} \sum_{j \in S_l} \check{p}_j$. Hence, 
$v_l(S_l) - \sum_{j \in S_l \setminus T}\check{p}'_j = v_l(S_l) - \sum_{j \in S_l \setminus T'}\check{p}'_j \geq v_l(S_l) - \sum_{j \in S_l \setminus T'} \check{p}_j \ge v_l(T') - \sum_{j \in T' \setminus S_l}\hat{p}_j = v_l(T) - \sum_{j \in T \setminus S_l}\hat{p}_j$, where 
the first equality follows from the fact that $|S_l \setminus T| = |S_l \setminus T'|$ and that $\check{p}'_j = \check{p}'_{j'}$ for every $j,j' \in S_l$.
The last inequality follows from Inequality (\ref{eqn:2pe}), and the last equality follows from the definition of $T'$ and the fact that the function $v$ is symmetrical. That is, Inequality (\ref{eqn:2pe_ident_1}) holds for buyer $l$ and every set $T \subseteq [m]$, and hence $(\mathbf{S}, \mathbf{\hat{p}}, \mathbf{\check{p}'})$ is a 2PE for $\vals$.
Finally,  $D(\mathbf{S}, \mathbf{\hat{p}}, \mathbf{\check{p}'}) = D(\mathbf{S}, \mathbf{\hat{p}}, \mathbf{\check{p}})$ follows directly from the definition of $D$ and the definition of $\check{p}'_j$.

To show that $(\mathbf{S}, \mathbf{\hat{p}'}, \mathbf{\check{p}})$ is a 2PE, it remains to show that (\ref{eqn:2pe_ident_2}) holds for every buyer $i \neq l$.
Let $X$ be the set of the 
$x$ items with the lowest $\hat{p}$ in $S_l$, where $x = \mid S_l \cap T \mid$, and
let $T' = X \cup (T \setminus S_l)$. Notice that $T' \cap S_l = X$ and therefore, $\sum_{j \in T' \cap S_l} \hat{p}_j \le \sum_{j \in T' \cap S_l} \hat{p}'_j = x \cdot \frac{1}{\mid S_l \mid} \sum_{j \in S_l} \hat{p}_j$. Hence, for every buyer $i \neq l$ we have, 
$v_i(S_i) - \sum_{j \in S_i \setminus T}\check{p}_j = 
v_i(S_i) - \sum_{j \in S_i \setminus T'}\check{p}_j \geq v_i(T') - \sum_{j \in T' \setminus S_i} \hat{p}_j = 
v_i(T') - \sum_{j \in T' \setminus (S_i \cup S_l)} \hat{p}_j - \sum_{j \in T' \cap S_l} \hat{p}_j \ge
v_i(T') - \sum_{j \in T' \setminus (S_i \cup S_l)} \hat{p}_j - \sum_{j \in T' \cap S_l} \hat{p}'_j = 
v_i(T) - \sum_{j \in T \setminus (S_i \cup S_l)} \hat{p}_j - \sum_{j \in T \cap S_l} \hat{p}'_j = v_i(T) - \sum_{j \in T \setminus S_i} \hat{p}'_j$,
where the first equality follows from the definition of $T'$,
the first inequality follows from Inequality (\ref{eqn:2pe}), the third equality is due to the definition of $T'$ and the fact that the function $v$ is symmetrical, and the last equality is derived form the fact that $\hat{p}'_j = \hat{p}_j$ for every item $j \notin S_l$.
We showed that, $v_i(S_i) - \sum_{j \in S_i \setminus T}\check{p}_j \ge v_i(T) - \sum_{j \in T \setminus S_i} \hat{p}'_j$ for every buyer $i$ and every set $T \subseteq [m]$, and hence $(\mathbf{S}, \mathbf{\hat{p}'}, \mathbf{\check{p}})$ is a 2PE for $\vals$.
Finally,  $D(\mathbf{S}, \mathbf{\hat{p}}', \mathbf{\check{p}}) = D(\mathbf{S}, \mathbf{\hat{p}}, \mathbf{\check{p}})$ follows directly from the definition of $D$ and the definition of $\hat{p}'_j$. 
\end{proof}

The complete proofs of the following propositions appear in Appendix \ref{miss_proof_2pe_id}.

The following proposition gives \emph{necessary and sufficient} conditions for the utility maximization property of a U-2PE.

\begin{proposition}
\label{prop:u2pe_necess_cond}
Consider a symmetric valuation profile $\vals$ and a triplet $(\mathbf{S}, \mathbf{\hat{p}}, \mathbf{\check{p}})$, s.t. 
for every item $j \in [m]$, $\check{p_j} \le \hat{p_j}$
and for every buyer $i \in [n]$ and every items $j,j' \in S_i$, $\hat{p}_j = \hat{p}_{j'}$ and $\check{p}_j = \check{p}_{j'}$. 
Then, the following conditions are necessary and sufficient for utility maximization of a U-2PE:

\begin{enumerate}

\item
\label{um_cond_1}
$\check{p}^{(i)} \le \min_{i' \in [n]} \hat{p}^{(i')}$, for every $i \in [n]$.

\item
\label{um_cond_2}
$\check{p}^{(i)} \le \overleftarrow{\Delta}_{v_i}(|S_i|)$, for every $i \in [n]$.

\item
\label{um_cond_3}
$\sum_{i' \neq i} |T \cap S_{i'}| \cdot \hat{p}^{(i')} \ge v_i(|T|) - v_i(|S_i|)$, for every $i \in [n]$ and every $T \subseteq [m]$ s.t.
$T \supset S_i$.

\end{enumerate}

\end{proposition}

Given Proposition \ref{prop:u2pe_necess_cond}, we can now specify simple sufficient conditions for
U-2PE in market with identical items.

\begin{proposition}
\label{prop:2pe_ecp_n_suff_cond}
Consider a symmetric valuation profile $\vals$ and let $(\mathbf{S}, \mathbf{\hat{p}}, \mathbf{\check{p}})$ be a triplet satisfying the following conditions for every buyer $i \in [n]$:

\begin{enumerate}
\item
\label{u2pe_cond_1}
For every items $j,j' \in S_i$, $\hat{p}_j = \hat{p}_{j'}$ and $\check{p}_j = \check{p}_{j'}$. Let $\hat{p}^{(i)}$ and $\check{p}^{(i)}$ denote these prices, respectively.

\item
\label{u2pe_cond_2}
$\check{p}^{(i)} \le \min_{i' \in [n]} \hat{p}^{(i')}$.

\item
\label{u2pe_cond_3}
$\check{p}^{(i)} \le \overleftarrow{\Delta}_{v_i}(|S_i|)$.

\item
\label{u2pe_cond_4}
$\hat{p}^{(i)} \ge \max_{i' \neq i} \set{ \overrightarrow{\Delta}_{v_{i'}}(|S_{i'}|)}$. 

\item
\label{u2pe_cond_5}
All items are allocated.
\end{enumerate}

Then, $(\mathbf{S}, \mathbf{\hat{p}}, \mathbf{\check{p}})$ is a U-2PE for $\vals$.
\end{proposition}

Notice that for the case of two buyers, condition (\ref{u2pe_cond_4}) of Proposition \ref{prop:2pe_ecp_n_suff_cond} is identical to condition (\ref{um_cond_3}) of Proposition \ref{prop:u2pe_necess_cond}, and therefore the conditions specified in Proposition \ref{prop:2pe_ecp_n_suff_cond} are also necessary conditions.




\section{Discrepancy in Markets with Identical Subadditive Buyers}
\label{sec:id_costumers}
In this section we establish the existence of 2PEs with small discrepancy for markets with identical items and identical subadditive buyers. 

We first show that every market with identical items and 2 identical subadditive buyers admits a 2PE with discrepancy of at most $2$.


\begin{theorem}
\label{thm:sa_n2_q2}
Every market with 2 identical subadditive symmetric valuations admits a U-2PE $(\mathbf{S}, \mathbf{\hat{p}}, \mathbf{0})$ with discrepancy of at most 2.
\end{theorem}

\begin{proof}
Consider a valuation $v$ and an allocation $\mathbf{S}$. Let $k_1 = |S_1|$ and $k_2 = |S_2|$.
For every $j \in S_1$, let $\hat{p}_j = \overrightarrow{\Delta}_{v}(k_2) = \hat{p}^{(1)}$ and for every $j \in S_2$, let $\hat{p}_j = \overrightarrow{\Delta}_{v}(k_1) = \hat{p}^{(2)}$. As $\check{p}_j = 0$ for every $j \in [m]$,  and $\overrightarrow{\Delta}_{v}(k) \ge 0$, $\overleftarrow{\Delta}_{v}(k) \ge 0$ for every $0 \le k < m$, the triplet $(\mathbf{S}, \mathbf{\hat{p}}, \mathbf{0})$ satisfies all the conditions in Proposition \ref{prop:2pe_ecp_n_suff_cond} and hence it is a U-2PE. 

It holds that
\begin{eqnarray}
\label{eqn:k1_k2_q}
D(\mathbf{S}, \mathbf{\hat{p}}, \mathbf{0}) =
\frac{\sum_{j \in [m]} \hat{p}_j }{SW(\mathbf{S}, \vals)} = \frac{k_1 \cdot \hat{p}^{(1)} + k_2 \cdot \hat{p}^{(2)}}{v(k_1) + v(k_2)} = \frac{k_1 \cdot \overrightarrow{\Delta}_{v}(k_2) + k_2 \cdot \overrightarrow{\Delta}_{v}(k_1)}{v(k_1) + v(k_2)}
\end{eqnarray}
Our goal is to show that there exists a pair ($k_1$,$k_2$), s.t $0 \le k_1, k_2 \le m$, $k_1+k_2 = m$ and $D(\mathbf{S}, \mathbf{\hat{p}}, \mathbf{0}) \le 2$. Therefore, it suffices to prove that there exists at least one \emph{$2-$good pair}, ($k_1$,$k_2$), i.e., that both $k_1$ and $k_2$ are $2-$good elements.
Indeed, for such a pair,

\begin{eqnarray}
\label{eqn:k1_k2_q_1}
D(\mathbf{S}, \mathbf{\hat{p}}, \mathbf{0}) = \frac{k_1 \cdot \overrightarrow{\Delta}_{v}(k_2) + k_2 \cdot \overrightarrow{\Delta}_{v}(k_1)}{v(k_1) + v(k_2)} \le \frac{(k_1 + k_2) \cdot 2  \frac{v(m)}{m}}{v(k_1) + v(k_2)} \le \frac{2 v(m)}{v(m)}= 2.
\end{eqnarray}
where the first inequality follows from the fact that both $k_1$ and $k_2$ are $2-$good elements and the last inequality is due to subadditivity.
Notice that there are exactly $m + 1$ pairs that satisfy $k_1, k_2 \ge 0$ and $k_1+k_2 = m$.
According to Lemma \ref{lem:c-bad}, there are at most $m - \lfloor \frac{m}{2} \rfloor - 1 < \frac{m}{2}$ 
integers $k_1 \in \set{0,1, \ldots, m-1}$
which are $2-$bad, i.e. with
$\overrightarrow{\Delta}_{v}(k_1) > 2 \cdot \frac{v(m)}{m}$.
Similarly, there are at most $m - \lfloor \frac{m}{2} \rfloor - 1 < \frac{m}{2}$ 
integers $k_2 \in \set{0,1, \ldots, m-1}$
which are $2-$bad, i.e. with
$\overrightarrow{\Delta}_{v}(k_2) > 2 \cdot \frac{v(m)}{m}$.
Overall there are at most 
$2 \cdot (m - \lfloor \frac{m}{2} \rfloor - 1) < m$ $2-$bad elements, and since this number is strictly less than $m + 1$, which is the number of pairs, there exists at least one good pair that satisfies
$k_1, k_2 \ge 0$,
$k_1+k_2 = m$,
$\overrightarrow{\Delta}_{v}(k_1) \le 2 \cdot \frac{v(m)}{m}$ and
$\overrightarrow{\Delta}_{v}(k_2) \le 2 \cdot \frac{v(m)}{m}$, which concludes the proof.

\end{proof}

We next establish a lower bound on the discrepancy of a 2PE for 2 identical subadditive buyers. 

\begin{theorem}
\label{ex:d_138}
There exists a market with identical items and 2 identical subadditive buyers that admits no 2PE with discrepancy smaller than $1.3895$.
\end{theorem}

\begin{proof}
Consider a setting with $3461$ identical items and two identical buyers with valuation function as follows:

\begin{eqnarray*}
v(k) =  
	\begin{cases}
		0                                        & k = 0 \\
		\lfloor \frac{k-1}{30} \rfloor +1        & 0 < k \le 480 \\
		\lfloor \frac{k-481}{50} \rfloor +17     & 480 < k \le 2980 \\ 
		\lfloor \frac{k-2981}{30} \rfloor +67    & 2980 < k \le 3461 \\
	\end{cases} \\
\end{eqnarray*}


There are $3462$ possible allocations, in which the first buyer gets $k$ items and the second gets $3461-k$ items, where $0 \le k \le 3461$. 
We claim that the minimum discrepancy is achieved at $k = 1$ and $k = 3460$ with discrepancy of slightly above  $1.3895$.
This is proved using a computer program \cite{code}.

\end{proof}

Example \ref{thm:simple_ex} in Appendix \ref{lb_simple} gives a simpler instance showing a lower bound of $\frac{6}{5}$.


We now extend the result of Theorem \ref{thm:sa_n2_q2} to markets with an arbitrary number of identical subadditive buyers.

\begin{theorem}
\label{thm:sa_n_q25}
Every market with $n > 2$ identical subadditive symmetric valuations admits a U-2PE, $(\mathbf{S}, \mathbf{\hat{p}}, \mathbf{\check{p}})$, with discrepancy of at most  $\max \set{2, \frac{n+2}{n-1}} \le 2.5$.

\end{theorem}

To prove Theorem \ref{thm:sa_n_q25}, we present an algorithm that 
computes some allocation $(k_1, k_2, \ldots, k_n)$, and show in Lemma~\ref{lem:n_d25} that the obtained allocation is supported in a 2PE with discrepancy of at most $\max \set{2, \frac{n+2}{n-1}}$. 

Line~\ref{alg1:find_pair} in the algorithm refers to a \emph{$2-$good pair}. For two buyers $x, y \in [n]$ and integers $k_x,l_x,k_y,l_y,r \in [m]$, we say that a pair $(l_x, l_y)$ is  \emph{$2-$good} w.r.t. $r$, $k_x$ and $k_y$ if (i) $l_x, l_y \ge 0$ (ii) $l_x+l_y = r$, (iii) $\overrightarrow{\Delta}_{v_x}(k_{x} + l_x) \le 2 \cdot  \overrightarrow{\Delta}_{v_x}(k_{x})$, and (iiii) $\overrightarrow{\Delta}_{v_y}(k_{y} + l_y) \le 2 \cdot \overrightarrow{\Delta}_{v_y}(k_{y})$.

In the beginning, the algorithm allocates ``whole triangles" to buyers, each time allocating to the buyer with the highest max-forward slope, breaking ties in favor of buyers with smaller bundles. As buyer valuations are identical, this ensures that a triangle in $T_v$ is allocated to \emph{all} buyers before the next triangle in $T_v$ is allocated to any buyer.
If the number of remaining items, $r$, is less than the number of items in the selected buyer's triangle, i.e., there are not enough elements to allocate the whole triangle that was chosen, 
the algorithm allocates the $r$ remaining items to two buyers, s.t. the pair of  the max-forward slopes is a 2-good pair.
Note that the naive idea of allocating all the $r$ items to the selected buyer is potentially bad, because we have no control over the max-forward slope within the triangle.

\begin{algorithm}[H]
\caption{An algorithm for finding an allocation with discrepancy of at most 2.5 for identical buyers.}
\label{alg1:algorithm}
\textbf{Input}: $m,n,v$\\
\textbf{Output}: $(k_1, k_2, \ldots, k_n)$, s.t. $\sum_{i \in [n]} k_i = m$ and $k_i \ge 0$ for every $i \in [n]$
\begin{algorithmic}[1] 
\STATE Let $k_i=0$ for every $i \in [n]$
\STATE Let $r = m$
\WHILE{$r > 0$}
\STATE \label{alg1:choose_step} Let $X = \{i | i = argmax_{i' \in [n]} \set{\overrightarrow{\Delta}_{v}(k_{i'})}\}$ 
\STATE \label{alg1:choose_step1} Let $x = argmin_{i \in X} \set{k_i}$
\STATE Let $t \ge 1$ be the number of items in $x$'s current triangle.
\IF {$r \ge t$}  \label{alg1:if_step}
\STATE $k_x = k_x + t$
\STATE $r = r - t$
\ELSE \label{alg1:else_step}
\STATE Let $Y = \{i | i = argmax_{i' \in [n]\setminus x} \set{\overrightarrow{\Delta}_{v}(k_{i'})}\}$ 
\STATE \label{alg1:y_choose} Let $y = argmin_{i \in Y} \set{k_i}$
\STATE \label{alg1:find_pair} Find a $2-$good pair, $(l_x, l_y)$ w.r.t. $r$,$k_x$,$k_y$
\STATE $k_x = k_x + l_x$
\STATE $k_y = k_y + l_y$
\STATE $r = 0$
\ENDIF
\ENDWHILE
\STATE \textbf{return} $(k_1, k_2, \ldots, k_n)$
\end{algorithmic}
\end{algorithm}


The following lemma shows that every allocation that is obtained as an output of Algorithm \ref{alg1:algorithm} can be supported in a U$-$2PE with the desired discrepancy.
\begin{lemma}
\label{lem:n_d25}
Let $(k_1, k_2, \ldots, k_n)$ be an allocation returned by Algorithm \ref{alg1:algorithm}, and let $\mathbf{S} = (S_1, S_2, \ldots, S_n)$ be a an allocation satisfying $|S_i| = k_i$. There exist $\highprice$ and $\lowprice$ s.t. $(\mathbf{S}, \mathbf{\hat{p}}, \mathbf{\check{p}})$ is a U$-$2PE with discrepancy of at most $\max \set{2, \frac{n+2}{n-1}}$.
\end{lemma}

\begin{proof}

First note that if Algorithm \ref{alg1:algorithm} ends without going through step \ref{alg1:else_step}, then the $k_i$ of each buyer $i$ is located at the beginning of a triangle.
Let $x$ be the last buyer that has been chosen in step \ref{alg1:choose_step1} and let $\theta_x$ be the slope of buyer $x$ before entering step \ref{alg1:if_step}. 
Note that if $k_i > 0$, then at the last time that buyer $i$ was chosen in step \ref{alg1:choose_step1}, she had max forward slope of at least $\theta_x$. Therefore,
By Corollary \ref{cor:min-slope} and lemma \ref{lem:intgr-lem}, the SW is at least $m \cdot \theta_x$.
Moreover, by step \ref{alg1:choose_step} and Lemma \ref{lem:mon-tri-slopes}, for every $i \in [n]$, $\overrightarrow{\Delta}_{v_i}(k_{i}) \le \theta_x$.
Let  
$\hat{p}_j = \max_{i' \neq i} \set{\overrightarrow{\Delta}_{v_{i'}}(k_{i'})} \le \theta_x$ 
for every $j \in S_i$
and $\check{p}_j = 0$ for every item $j \in [m]$. 
It is easy to see that $(\mathbf{S}, \mathbf{\hat{p}}, \mathbf{0})$ satisfies all the conditions of Proposition \ref{prop:2pe_ecp_n_suff_cond}, and hence it is a U$-$2PE.
The discrepancy is, 
\begin{eqnarray*}
D(\mathbf{S}, \mathbf{\hat{p}}, \mathbf{0}) = \frac{\sum_{j \in [m]} \hat{p}_j}{SW} \le \frac{m \cdot \theta_x}{SW} \le \frac{m \cdot \theta_x}{m  \cdot \theta_x} = 1
\end{eqnarray*}
Note that if the $k_i$ of each buyer $i$ in Algorithm \ref{alg1:algorithm}'s output is located at the beginning of a triangle, then for every $i \in [n]$, $\overrightarrow{\Delta}_{v_i}(k_{i}) \le \theta_x$ and also $\overleftarrow{\Delta}_{v_i}(k_{i}) \ge \theta_x$. 
One can easily verify that for $\mathbf{p} = (\theta_x, \ldots, \theta_x)$, $(\mathbf{S}, \mathbf{p}, \mathbf{p})$ is a WE for the SM-closure valuation profile $\mathbf{\tilde{v}}$ and according to Theorem \ref{thm:we_over_ident}, it is also a WE for $\vals$.

Now assume Algorithm \ref{alg1:algorithm} goes through step \ref{alg1:else_step} before ending. Let $k'_i$ be the value of $k_i$ and $r'$ the value of $r$ when the algorithm enters step \ref{alg1:else_step}. We denote \emph{current triangle} of buyer $i$ as the triangle with $k'_i$ as its left most point and \emph{previous triangle} of buyer $i$ as the triangle with $k'_i$ as its right most point. Let us denote $T^x$ as buyer $x$ current triangle. We consider two cases. The first one is when there is another buyer, $y$, with the same current triangle, $T^x$. Then by Lemma \ref{lem:mon-tri-slopes}, buyer $y$ is chosen at step \ref{alg1:y_choose}. Let $\theta_y$ be the slope of buyer $y$ when she is chosen. Since both buyers current triangle is identical we have that $\theta_x = \theta_y = \theta$. By Claim \ref{clm:min-slope-in-tri} and step \ref{alg1:find_pair}, $\theta \le \overrightarrow{\Delta}_{v}(k'_{x} + l_x) \le 2 \cdot \theta$
and $\theta \le \overrightarrow{\Delta}_{v}(k'_{y} + l_y) \le 2 \cdot \theta$. By step \ref{alg1:choose_step} and Lemma \ref{lem:mon-tri-slopes}, the slopes of all buyers current triangles are at most $\theta$. Thus, we can set $\hat{p}_j = 2 \cdot \theta$ and $\check{p}_j = 0$ for every item $j \in [m]$. Once again, it is easy to see that  $(\mathbf{S}, \mathbf{\hat{p}}, \mathbf{\check{p}})$ satisfies all the conditions of Claim \ref{prop:2pe_ecp_n_suff_cond}, and hence it is a U$-$2PE. Next, we bound the SW. Let $i$ be a buyer with previous triangle $T^x$. By step \ref{alg1:choose_step} and Lemma \ref{lem:mon-tri-slopes}, the minimum max-forward-slope up to $k_i$ is $\theta$. By Lemma \ref{lem:intgr-lem} we have that $v_i(k_i) \ge k_i \cdot \theta$. Now, let $i$ be a buyer with current triangle $T^x$, then by Lemma \ref{lem:mon-tri-slopes}, $v_i(k_i) \ge k_i \cdot \theta' \ge k_i \cdot \theta $ where $\theta'$ is buyer $i$ previous triangle slope. We can now bound the SW by $SW(\mathbf{S}, \vals) \ge m \cdot \theta$. Concluding, $$D(\mathbf{S}, \mathbf{\hat{p}}, \mathbf{\check{p}}) =\frac{\sum_{j \in [m]} \hat{p}_j - \sum_{j \in [m]} \check{p}_j}{SW(\mathbf{S}, \vals)} \le \frac{2 \cdot m \cdot \theta}{m \cdot \theta} = 2$$
The second case is when buyer $x$ is the only buyer with $T^x$ as its current triangle (i.e., all other buyers previous triangle is $T^x$). By Claim \ref{clm:min-slope-in-tri} and step \ref{alg1:find_pair}, $\theta_x \le \overrightarrow{\Delta}_{v}(k'_{x} + l_x) \le 2 \cdot \theta_x$
and $\theta_y \le \overrightarrow{\Delta}_{v}(k'_{y} + l_y) \le 2 \cdot \theta_y \le 2 \cdot \theta_x$. Since for every $i \ne x,y$ we have that $\overrightarrow{\Delta}_{v}(k_i) = \theta_y$ we can set $\hat{p}_j = 2 \cdot \theta_x$ for every item $j \in [m]$. Note that for every $i,i' \in [n] \setminus \{x\}$ $k'_i = k'_{i'} = z$. Namely, before the last step of the algorithm every buyer except from $x$ has $z$ items and $r' < z$. Thus, $$m - r' = (n-1) z > m \frac{n-1}{n},$$ where the last inequality is due to the fact that $n \cdot z > m$. Since for every $i \ne x,y$ we have that $\overleftarrow{\Delta}_{v}(k_i) = \theta_x$ we can set $\check{p_j} = \theta_x$ for every item $j \in [m] \setminus S_x \cup S_y$ and $\check{p_j} = 0$ for every item $j \in S_x \cup S_y$ and get that $\sum_{j \in [m]} \check{p_j} = (n-2) \cdot z \cdot \theta_x > \frac{n-2}{n} \cdot m \cdot \theta_x$. As for the SW, since the minimum max-forward-slope before entering step \ref{alg1:else_step} is at least $\theta_x$, from Lemma \ref{lem:intgr-lem} we have that $SW(\mathbf{S}, \vals) \ge (m-r') \cdot \theta_x > \frac{n-1}{n} \cdot m \cdot \theta_x$. Combining it all together, $$D(\mathbf{S}, \mathbf{\hat{p}}, \mathbf{\check{p}}) = \frac{\sum_{j \in [m]} \hat{p}_j - \sum_{j \in [m]} \check{p}_j}{SW(\mathbf{S}, \vals)} \le \frac{2 \cdot \theta_x \cdot m - \frac{n-2}{n} \cdot m \cdot \theta_x}{\frac{n-1}{n} \cdot m \cdot \theta_x} = \frac{n+2}{n-1}$$
\end{proof}

It now remains to show that there always exists a $2-$good pair in line \ref{alg1:find_pair} of Algorithm \ref{alg1:algorithm}. This is established in the following lemma, which concludes the proof of Theorem~\ref{thm:sa_n_q25}.

\begin{lemma}
\label{lem:sa_n_q25_2good_pair}
For every two buyers $x,y$ in line \ref{alg1:find_pair} of Algorithm \ref{alg1:algorithm}, there exists a $2-$good pair $(l_x, l_y)$ with respect to $r$,$k_x$, and $k_y$.
\end{lemma}

\begin{proof}
Let $k'_x$, $k'_y$, $r'$ and $t'$ be the respective values of $k_x$, $k_y$, $r$ and $t$ when the algorithm enters step \ref{alg1:find_pair}. It holds that $k'_x$ and $k'_y$ are located at the beginning of a triangle, with slopes $\theta_x$ and $\theta_y$, respectively. 
Hence, we need to show that there exists a pair $(l_x, l_y)$ s.t.: 
$l_x, l_y \ge 0$,
$l_x+l_y = r'$, 
$\overrightarrow{\Delta}_{v}(k'_{x} + l_x) \le 2 \cdot \theta_x$, and 
$\overrightarrow{\Delta}_{v}(k'_{y} + l_y) \le 2 \cdot \theta_y$.

Let $T^x$ be buyer $x$'s current triangle. We know that its length is $t'$ and that $t' > r'$. We would like to "truncate" $T^x$ after $r'+1$ points and consider the max-forward-slopes of the truncated triangle. 
Let $v': [r'+1] \rightarrow R^+$ be the following monotone set function: 
$$v'(k) = 
	\begin{cases}
		v(k + k'_x) - v(k'_x)        & 0 \le k \le r' \\
		(r'+1) \cdot \theta_x     & k = r'+1.
	\end{cases}$$

Notice that $v'$ consists of one triangle, $T^{x'}$,  with length $r'+1$ and slope $\theta_x$. Geometrically, $T^{x'}$ is identical to $T^{x}$ in its first $r'$ points, but the $(r'+1)^{th}$ point of $T^{x'}$ is higher than the $(r'+1)^{th}$ point of $T^{x}$, since every point in $T^{x}$ is located strictly below the hypotenuse. Hence, for every $0  \le k \le r'$, $\overrightarrow{\Delta}_{v'}(k) \ge \overrightarrow{\Delta}_{v}(k'_x + k)$

Consider now buyer $y$. Let $v'': [r'+1] \rightarrow R^+$ be the following monotone set function:
$$v''(k) = 
	\begin{cases}
		v(k + k'_y) - v(k'_y)        & 0 \le k \le r' \\
		(r'+1) \cdot \theta_y     & k = r'+1.
	\end{cases}$$
Similar to the arguments of buyer $x$, 
$\overrightarrow{\Delta}_{v''}(k) \ge \overrightarrow{\Delta}_{v}(k'_y + k)$, for every $0  \le k \le r'$.

Therefore, it suffices to show that there exists a pair $(l_x, l_y)$ s.t.: $l_x, l_y \ge 0$,
$l_x+l_y = r'$, 
$\overrightarrow{\Delta}_{v'}(l_x) \le 2 \cdot \theta_x$, and 
$\overrightarrow{\Delta}_{v''}(l_y) \le 2 \cdot \theta_y$. Notice that there are exactly $r' + 1$ pairs that satisfy $l_x, l_y \ge 0$ and $l_x+l_y = r'$.
According to Lemma \ref{lem:c-bad}, there are at most $r'+1 - \lfloor \frac{1}{2} \cdot (r'+1) \rfloor - 1 = r'- \lfloor \frac{1}{2} \cdot (r'+1) \rfloor$ element in $v'$ which are $2-$bad, i.e. with max-forward-slope which is strictly more than $2 \cdot \theta_x$.
Similarly, there are at most $r'- \lfloor \frac{1}{2} \cdot (r'+1) \rfloor$ element in $v''$ which are $2-$bad, i.e. with max-forward-slope which is strictly higher than $2 \cdot \theta_y$.
Overall there are at most $2 \cdot (r'- \lfloor \frac{1}{2} \cdot (r'+1) \rfloor)$ $2-$bad elements in $v'$ and $v''$, and since this number is strictly less than $r' + 1$, which is the number of pairs, there exists at least one good pair that satisfies
$l_x, l_y \ge 0$,
$l_x+l_y = r'$, 
$\overrightarrow{\Delta}_{v}(k'_x + l_x) \le \overrightarrow{\Delta}_{v'}(l_x) \le 2 \cdot \theta_x$, and 
$\overrightarrow{\Delta}_{v}(k'_y + l_y) \le \overrightarrow{\Delta}_{v''}(l_y) \le 2 \cdot \theta_y$.
\end{proof}

\section{Discrepancy in Markets with Heterogeneous Subadditive Buyers}
\label{sec:het_costumers}


In this section we show that for every market with identical items and any number of subadditive buyers, there exists a 2PE with discrepancy of at most 6.

\begin{theorem}
\label{thm:sa_n_q6}
Every market with subadditive symmetric valuations admits a U$-$2PE, $(\mathbf{S}, \mathbf{\hat{p}}, \mathbf{0})$, with discrepancy of at most $6$.
\end{theorem}

To prove Theorem \ref{thm:sa_n_q6}, we present an algorithm that 
computes some allocation $(k_1, k_2, \ldots, k_n)$, and show in Lemma~\ref{lem:n_d6} that the obtained allocation is supported in a 2PE with discrepancy of at most $6$. 


Line~\ref{alg2:find_pair} in the algorithm refers to a \emph{$3-$good pair}. For two buyers $x, y \in [n]$ and integers $k_x,l_x,k_y,l_y,r \in [m]$, we say that a pair $(l_x, l_y)$ is  \emph{$3-$good} w.r.t. $r$, $k_x$ and $k_y$ if (i) $l_x, l_y \ge 0$: (ii) $l_x+l_y = r$, (iii) $\overrightarrow{\Delta}_{v_x}(k_{x} + l_x) \le 3 \cdot  \overrightarrow{\Delta}_{v_x}(k_{x})$, and (iiii) $\overrightarrow{\Delta}_{v_y}(k_{y} + l_y) \le 3 \cdot \overrightarrow{\Delta}_{v_y}(k_{y})$.

Similar to Algorithm \ref{alg1:algorithm},
the algorithm start by allocating ``whole triangles" to buyers, each time allocating to the buyer with the highest max-forward slope. 
When there are not enough elements to allocate the whole triangle that was chosen, 
the algorithm allocates the $r$ remaining items to two buyers s.t. the pair of  the max forward slopes is a 3-good pair.

%
%
%

\begin{algorithm}[H]
\caption{An algorithm for finding an allocation with discrepancy of at most 6 for heterogeneous buyers.}
\label{alg2:algorithm}
\textbf{Input}: $m,n,(v_1, v_2, \ldots, v_n)$\\
\textbf{Output}: $(k_1, k_2, \ldots, k_n)$, s.t. $\sum_{i \in [n]} k_i = m$ and $k_i \ge 0$ for every $i \in [n]$
\begin{algorithmic}[1] 
\STATE Let $k_i=0$ for every $i \in [n]$
\STATE Let $r = m$
\WHILE{$r > 0$}
\STATE \label{alg2:choose_step} Let $x = argmax_{i \in [n]} \set{\overrightarrow{\Delta}_{v_{i}}(k_{i})}$
\STATE Let $t \ge 1$ be the number of items in $x$'s current triangle.
\IF {$r \ge t$} \label{alg2:if_step}
\STATE $k_x = k_x + t$
\STATE $r = r - t$
\ELSE \label{alg2:else_step}
\STATE \label{alg2:y_choose} Let $y = argmax_{i \in [n] \backslash x} \set{\overrightarrow{\Delta}_{v_{i}}(k_{i})}$
\STATE \label{alg2:find_pair} Find a $3-$good pair, $(l_x, l_y)$ w.r.t. $r$,$k_x$,$k_y$, s.t. $l_x \ge \frac{r}{2}$
\STATE $k_x = k_x + l_x$
\STATE $k_y = k_y + l_y$
\STATE $r=0$
\ENDIF
\ENDWHILE
\STATE \textbf{return} $(k_1, k_2, \ldots, k_n)$
\end{algorithmic}
\end{algorithm}

Given the output $(k_1, k_2, \ldots, k_n)$ of Algorithm \ref{alg2:algorithm}, let $\mathbf{S} = (S_1, S_2, \ldots, S_n)$ be an allocation that satisfies $|S_i| = k_i$ and let  
$\hat{p}_j = \max_{i' \neq i} \set{\overrightarrow{\Delta}_{v_{i'}}(k_{i'})}$ 
for every $j \in S_i$
and $\check{p}_j = 0$ for every item $j \in [m]$. 
It is easy to see that $(\mathbf{S}, \mathbf{\hat{p}}, \mathbf{0})$ satisfies all the conditions of Proposition \ref{prop:2pe_ecp_n_suff_cond}, and hence it is a U$-$2PE.

The following lemma shows that  $(\mathbf{S}, \mathbf{\hat{p}}, \mathbf{0})$ has the desired discrepancy.
\begin{lemma}
\label{lem:n_d6}
The discrepancy of $(\mathbf{S}, \mathbf{\hat{p}}, \mathbf{0})$ is at most $6$.
\end{lemma}

\begin{proof}
First note that if Algorithm \ref{alg2:algorithm} ends without going through step \ref{alg2:else_step}, then the $k_i$ of each buyer $i$ is located at the beginning of a triangle.
Let $x$ be the last buyer that has been chosen in step \ref{alg2:choose_step} and let $\theta_x$ be the slope of buyer $x$ before entering step \ref{alg2:if_step}. By step \ref{alg2:choose_step} and Lemma \ref{lem:mon-tri-slopes}, for every $i \in [n]$, $\overrightarrow{\Delta}_{v_i}(k_{i}) \le \theta_x$ and therefore $\hat{p}_j \le \theta_x$ for every $j \in [m]$.
Note that if $k_i > 0$, then at the last time that buyer $i$ was chosen in step 4, she had max forward slope of at least $\theta_x$. Therefore,
By Corollary \ref{cor:min-slope} and lemma \ref{lem:intgr-lem}, the SW is at least $m \cdot \theta_x$.
Therefore, 
\begin{eqnarray*}
D(\mathbf{S}, \mathbf{\hat{p}}, \mathbf{0}) = \frac{\sum_{j \in [m]} \hat{p}_j}{SW} \le \frac{m \cdot \theta_x}{SW} \le \frac{m \cdot \theta_x}{m  \cdot \theta_x} = 1
\end{eqnarray*}
Note that if the $k_i$ of each buyer $i$ in Algorithm \ref{alg2:algorithm}'s output is located at the beginning of a triangle, then for every $i \in [n]$, $\overrightarrow{\Delta}_{v_i}(k_{i}) \le \theta_x$ and also $\overleftarrow{\Delta}_{v_i}(k_{i}) \ge \theta_x$. 
One can easily verify that for $\mathbf{p} = (\theta_x, \ldots, \theta_x)$, $(\mathbf{S}, \mathbf{p}, \mathbf{p})$ is a WE for the SM-closure valuation profile $\mathbf{\tilde{v}}$ and according to Theorem \ref{thm:we_over_ident}, it is also a WE for $\vals$. 

We now assume that Algorithm \ref{alg2:algorithm} goes through step \ref{alg2:else_step} before ending. 
Let $k'_i$ be the value of $k_i$ and $r'$ the value of $r$ when the algorithm enters step \ref{alg2:else_step}.  We denote \emph{current triangle} of buyer $i$ as the triangle with $k'_i$ as its left most point and \emph{previous triangle} of buyer $i$ as the triangle with $k'_i$ as its right most point.
At this stage, the algorithm allocated $m-r'$ items, i.e. $\sum_{i \in n} k'_i = m-r'$.
By step \ref{alg2:choose_step} and Lemma \ref{lem:mon-tri-slopes}, the slope of all buyers previous triangle is at least $\theta_x$, and the slope of all buyers current triangle is at most $\theta_x$. 

Consider the stage in which Algorithm \ref{alg2:algorithm} enters step \ref{alg2:else_step} and let $\theta_y$ be the slope of buyer $y$ when she is chosen. Since $k'_i$ is located at the beginning of a triangle, by Lemma \ref{lem:tri-slope}, the slope of the current triangle of buyer $i$ equals to the max-forward-slope at that point, $\overrightarrow{\Delta}_{v_{i}}(k'_{i})$.
Notice that by step \ref{alg2:y_choose}, for every  buyer $i \notin \set{x,y}$, $ \overrightarrow{\Delta}_{v_{i}}(k'_{i}) \le \overrightarrow{\Delta}_{v_{y}}(k'_{y}) = \theta_y \le \theta_x$.
By Claim \ref{clm:min-slope-in-tri} and step \ref{alg2:find_pair}, $\theta_x \le \overrightarrow{\Delta}_{v_x}(k'_{x} + l_x) \le 3 \cdot \theta_x$
and $\theta_y \le \overrightarrow{\Delta}_{v_y}(k'_{y} + l_y) \le 3 \cdot \theta_y$.
Therefore, after the algorithm ends, for every $i \neq x$ and every $j \in S_i$, $\hat{p}_j = \max_{i' \neq i} \set{\overrightarrow{\Delta}_{v_{i'}}(k_{i'})} = \overrightarrow{\Delta}_{v_x}(k_{x}) = \overrightarrow{\Delta}_{v_x}(k'_{x} + l_x) \le 3 \cdot \theta_x$, and for every $j \in S_x$, $\hat{p}_j = \max_{i' \neq x} \set{\overrightarrow{\Delta}_{v_{i'}}(k_{i'})} = \overrightarrow{\Delta}_{v_y}(k_{y}) = \overrightarrow{\Delta}_{v_y}(k'_{y} + l_y) \le 3 \cdot \theta_y \le 3 \cdot \theta_x$. Hence,

\begin{eqnarray}
\label{eqn:d6_nom}
\sum_{j \in [m]} \hat{p}_j \le m \cdot 3 \cdot \theta_x
\end{eqnarray}

By Corollary \ref{cor:min-slope} and lemma \ref{lem:intgr-lem}, the SW of the $m-r'$ items that were allocated up to step \ref{alg2:else_step} is at least $(m-r') \cdot \theta_x$.
At step \ref{alg2:find_pair} we add at least $\frac{r'}{2}$ items to buyer $x$ and since $k'_x + l_x$ are located somewhere in buyer $x$'s current triangle, with slope $\theta_x$, the added SW of the $l_x$ items is, by lemma \ref{lem:intgr-lem}, is at least $\frac{r'}{2}  \cdot \theta_x$. Hence, 
\begin{eqnarray}
\label{eqn:d6_denom}
SW(\mathbf{S}, \vals) \ge (m-r'+\frac{r'}{2})  \cdot \theta_x = (m-\frac{r'}{2})  \cdot \theta_x \ge \frac{m}{2}  \cdot \theta_x, 
\end{eqnarray}
where the last inequality is due to the fact that $r' \le m$.

Putting it all together, we get:
\begin{eqnarray*}
D(\mathbf{S}, \mathbf{\hat{p}}, \mathbf{0}) = \frac{\sum_{j \in [m]} \hat{p}_j}{SW(\mathbf{S}, \vals)} \le \frac{m \cdot 3 \cdot \theta_x}{SW(\mathbf{S}, \vals)} \le \frac{m \cdot 3 \cdot \theta_x}{\frac{m}{2}  \cdot \theta_x} = 6
\end{eqnarray*}
where the first and second inequalities are due to Inequality (\ref{eqn:d6_nom}) Inequality (\ref{eqn:d6_denom}), respectively.
\end{proof}

To conclude the proof of Theorem~\ref{thm:sa_n_q6} it remains to establish the existence of a $3-$good pair that satisfies the condition in line \ref{alg2:find_pair} of Algorithm~\ref{alg2:algorithm}.

\begin{lemma}
\label{lem:sa_n_q6_3good_pair}
For every two buyers $x,y$ in line~\ref{alg2:find_pair} of Algorithm \ref{alg2:algorithm}, there exists a $3-$good pair $(l_x, l_y)$ with respect to $r$, $k_x$, and $k_y$ such that $l_x \ge \frac{r}{2}$.
\end{lemma}

\begin{proof}
Let $k'_x$, $k'_y$, $r'$ and $t'$ be the respective values of $k_x$, $k_y$, $r$ and $t$ when the algorithm enters step \ref{alg2:find_pair}. It holds that $k'_x$ and $k'_y$ are located at the beginning of a triangle, with slopes $\theta_x$ and $\theta_y$, respectively. 
Hence, we need to show that there exists a pair $(l_x, l_y)$ s.t.: 
$l_y \ge 0$,
$l_x \ge \frac{r'}{2}$, 
$l_x+l_y = r'$, 
$\overrightarrow{\Delta}_{v_x}(k_{x} + l_x) \le 3 \cdot \theta_x$, and 
$\overrightarrow{\Delta}_{v_y}(k_{y} + l_y) \le 3 \cdot \theta_y$.

Let $T^x$ be buyer $x$'s current triangle. We know that its length is $t'$ and that $t' > r'$. We would like to "truncate" $T^x$ after $r'+1$ points and consider the max-forward-slopes of the truncated triangle. 
Let $v'_x: [r'+1] \rightarrow R^+$ be the following monotone set function:
$$v'_x(k) = 
	\begin{cases}
		v_x(k + k'_x) - v_x(k'_x)        & 0 \le k \le r' \\
		(r'+1) \cdot \theta_x     & k = r'+1.
	\end{cases}$$
Notice that $v'_x$ consists of one triangle, $T^{x'}$,  with length $r'+1$ and slope $\theta_x$. Geometrically, $T^{x'}$ is identical to $T^{x}$ in its first $r'$ points, but the $(r'+1)^{th}$ point of $T^{x'}$ is higher than the $(r'+1)^{th}$ point of $T^{x}$, since every point in $T^{x}$ is located strictly below the hypotenuse. Hence, for every $0  \le k \le r'$, $\overrightarrow{\Delta}_{v'_{x}}(k) \ge \overrightarrow{\Delta}_{v_{x}}(k'_x + k)$

Consider now buyer $y$'s current triangle. It is followed by other triangles, each with slope at most $\theta_y$ (see Lemma \ref{lem:mon-tri-slopes}).
Let $v'_y: [\lfloor \frac{r'}{2} \rfloor +1] \rightarrow R^+$ be the following monotone set function:
$$v'_y(k) = 
	\begin{cases}
		v_y(k + k'_y) - v_y(k'_y)        & 0 \le k \le \lfloor \frac{r'}{2} \rfloor \\
		(\lfloor \frac{r'}{2} \rfloor+1) \cdot \theta_y     & k = \lfloor \frac{r'}{2} \rfloor+1.
	\end{cases}$$
Notice that $v'_y$ consists of one triangle with length of $\lfloor \frac{r'}{2} \rfloor+1$ and slope $\theta_y$. 
Similar to the arguments of buyer $x$, 
$\overrightarrow{\Delta}_{v'_{y}}(k) \ge \overrightarrow{\Delta}_{v_{y}}(k'_y + k)$, for every $0  \le k \le \lfloor \frac{r'}{2} \rfloor$.

Therefore, it suffices to show that there exists a pair $(l_x, l_y)$ s.t.: 
$l_y \ge 0$,
$l_x \ge \frac{r'}{2}$, 
$l_x+l_y = r'$, 
$\overrightarrow{\Delta}_{v'_x}(l_x) \le 3 \cdot \theta_x$, and 
$\overrightarrow{\Delta}_{v'_y}(l_y) \le 3 \cdot \theta_y$. Notice that $l_x$ can get the values $\lceil \frac{r'}{2} \rceil, \lceil \frac{r'}{2} \rceil + 1, \ldots, r$ and $l_y$ can get the values $\lfloor \frac{r'}{2} \rfloor, \lfloor \frac{r'}{2} \rfloor - 1, \ldots, 0$. Hence, there are exactly $\lfloor \frac{r'}{2} \rfloor + 1$ pairs that satisfy $l_y \ge 0$, $l_x \ge \frac{r'}{2}$, and $l_x+l_y = r'$.
According to Lemma \ref{lem:c-bad}, there are at most $r'+1 - \lfloor \frac{2}{3} \cdot (r'+1) \rfloor - 1 = r'- \lfloor \frac{2}{3} \cdot (r'+1) \rfloor$ elements in $v'_x$ which are $3-$bad, i.e. with max-forward-slope which is strictly more than $3 \cdot \theta_x$.
Similarly, there are at most $\lfloor \frac{r'}{2} \rfloor+1 - \lfloor \frac{2}{3} \cdot (\lfloor \frac{r'}{2} \rfloor+1) \rfloor - 1 = \lfloor \frac{r'}{2} \rfloor - \lfloor \frac{2}{3} \cdot (\lfloor \frac{r'}{2} \rfloor+1) \rfloor$ elements in $v'_y$ which are $3-$bad, i.e. with max-forward-slope which is strictly higher than $3 \cdot \theta_y \le 3 \cdot \theta_x$.
Overall there are at most $r'- \lfloor \frac{2}{3} \cdot (r'+1) \rfloor + \lfloor \frac{r'}{2} \rfloor - \lfloor \frac{2}{3} \cdot (\lfloor \frac{r'}{2} \rfloor+1) \rfloor$ $3-$bad elements in $v'_x$ and $v'_y$, and since this number is strictly less than $\lfloor \frac{r'}{2} \rfloor + 1$, which is the number of pairs, there exists at least one good pair that satisfies $l_y \ge 0$, $l_x \ge \frac{r'}{2}$, 
$l_x+l_y = r'$, 
$\overrightarrow{\Delta}_{v_{x}}(k'_x + l_x) \le \overrightarrow{\Delta}_{v'_x}(l_x) \le 3 \cdot \theta_x$, and 
$\overrightarrow{\Delta}_{v_{y}}(k'_y + l_y) \le \overrightarrow{\Delta}_{v'_y}(l_y) \le 3 \cdot \theta_y$.
\end{proof}

\section{WE in Markets with Identical Items}
\label{WEOverIdenticalItems}


In markets with identical items, one can restrict attention to WE in which all prices are equal. Indeed, if there are at least two buyers who are allocated, then it is clear. Otherwise, simply replace all prices by their average (see Lemma \ref{lem:2pe_ident_same_customer_low_prices}). A WE with a single price $p$ and allocation $\allocs$ is denoted by $(\mathbf{S},p)$.

The following theorem establishes necessary and sufficient conditions for the existence of a WE in markets with identical items. 

\begin{theorem}
\label{thm:we_over_ident}
Let $\mathbf{v} = (v_1, \ldots, v_n)$ be a symmetric valuation profile. 
Let $\tilde{v_i}$ be the SM-closure of $v_i$ for every $i \in [n]$, and let $\mathbf{\tilde{v}} = (\tilde{v_1}, \ldots, \tilde{v_n})$. 
$(\allocs, p)$ is a WE for valuation profile $\mathbf{v}$  if and only if $(\allocs, p)$ is a WE for valuation profile $\mathbf{\tilde{v}}$ and  $|S_i| \in I_{v_i}$ for every $i \in [n]$.
\end{theorem} 

\begin{proof} [Proof of Theorem \ref{thm:we_over_ident}]

'Only if' direction: 
Assume that $(\allocs, p)$ is a WE for $\mathbf{\tilde{v}}$ and $|S_i| \in I_{v_i}$ for every $i \in [n]$. 
Let $\mathbf{p} = (p_1, \ldots, p_m)$, where $p_j = p$, for every $j \in [m]$ . Since $(\allocs, p)$ is a WE, $(\mathbf{S}, \mathbf{{p}}, \mathbf{{p}})$ is a 2PE and 
as all prices are equal, 
it is also a U-2PE for $\mathbf{\tilde{v}}$. 
By condition (\ref{um_cond_2}) of Proposition \ref{prop:u2pe_necess_cond}, we have that $p \le \min_{i \in [n]} \set{\overleftarrow{\Delta}_{\tilde{v}_i}(|S_i|)}$.
Consider condition (\ref{um_cond_3}) of Proposition \ref{prop:u2pe_necess_cond} and let $t = |T|$, $t_{i'} = |T \cap S_{i'}|$ and $t_i = k_i = |S_i| < t$, $\sum_{i' \neq i} t_{i'} = t - k_i$. 
We get that, $p \ge \frac{\tilde{v}_i(t) - \tilde{v}_i(k_i)}{t - k_i}$
for every $i \in [n]$ and every $k_i < t \le m$. That is, $p \ge \max_{i \in [n]}\set{\overrightarrow{\Delta}_{\tilde{v}_i}(|S_{i}|)}$.  
Since $|S_i| \in I_{v_i}$ for every $i \in [n]$, by Lemma \ref{lem:tri-slope} we have that $\overrightarrow{\Delta}_{\tilde{v}_i}(|S_{i}|) = \overrightarrow{\Delta}_{v_i}(|S_{i}|)$ and $\overleftarrow{\Delta}_{\tilde{v}_i}(|S_{i}|) = \overleftarrow{\Delta}_{v_i}(|S_{i}|)$. Therefore $\max_{i \in [n]}\set{\overrightarrow{\Delta}_{v_i}(|S_{i}|)} \le p \le \min_{i \in [n]} \set{\overleftarrow{\Delta}_{v_i}(|S_i|)}$.
That is, conditions (\ref{u2pe_cond_1})-(\ref{u2pe_cond_5}) of Proposition \ref{prop:2pe_ecp_n_suff_cond} are satisfied. 
Hence, $(\mathbf{S}, \mathbf{{p}}, \mathbf{{p}})$ is a U-2PE for $\mathbf{v}$, and $(\allocs, p)$ is a WE for $\mathbf{v}$.

'If' direction: 
Assume that $(\allocs, p)$ is a WE for $\vals$. 
Similar to the 'Only if' direction, we can show that by conditions (\ref{um_cond_2}) and (\ref{um_cond_3}) of Proposition \ref{prop:u2pe_necess_cond},  $\max_{i \in [n]}\set{\overrightarrow{\Delta}_{v_i}(|S_{i}|)} \le p \le \min_{i \in [n]} \set{\overleftarrow{\Delta}_{v_i}(|S_i|)}$.
Let $T_{l}$ be a triangle and let $i$ be some buyer with bundle $S_i$, s.t. $|S_i| = i_l + k'$, and assume toward a contradiction that $k' > 0$. By Claim \ref{clm:min-slope-in-tri},  $\overrightarrow{\Delta}_{v_i}(i_l + k') > \overrightarrow{\Delta}_{v_i}(i_l) = \alpha_l$ and $\overleftarrow{\Delta}_{v_i}(i_l + k') < \overleftarrow{\Delta}_{v_i}(i_{l+1}) = \alpha_l$; namely, $p > \alpha_l$ and $p < \alpha_l$, a contradiction. Therefore,  $|S_i| \in I_{v_i}$ for every $i \in [n]$. By Lemma \ref{lem:tri-slope}, $\overrightarrow{\Delta}_{\tilde{v}_i}(|S_{i}|) = \overrightarrow{\Delta}_{v_i}(|S_{i}|)$ and $\overleftarrow{\Delta}_{\tilde{v}_i}(|S_{i}|) = \overleftarrow{\Delta}_{v_i}(|S_{i}|)$. 
Thus, $\max_{i \in [n]}\set{\overrightarrow{\Delta}_{\tilde{v}_i}(|S_{i}|)} \le p \le \min_{i \in [n]} \set{\overleftarrow{\Delta}_{\tilde{v}_i}(|S_i|)}$, 
i.e. conditions (\ref{u2pe_cond_1})-(\ref{u2pe_cond_5}) of Proposition \ref{prop:2pe_ecp_n_suff_cond} are satisfied. 
Hence, $(\mathbf{S}, \mathbf{{p}}, \mathbf{{p}})$ is a U-2PE for $\mathbf{\tilde{v}}$, and $(\allocs, p)$ is a WE for $\mathbf{\tilde{v}}$ and $|S_i| \in I_{v_i}$ for every $i \in [n]$.
\end{proof}

\section{Endowment Equilibrium}
\label{sec:EE}

In this section we present the notion of endowment equilibrium \cite{BDO18,EFF19}, and show its relation to 2PE. 
Discovered by Nobel laureate Richard Thaler, the endowment effect states that buyers tend to inflate the value of items they already own. 

This motivates the introduction of an {\em endowed valuation}, $v^{X}(\cdot):[m] \rightarrow \reals^+$, which assigns some real value to every bundle $Y \subseteq [m]$, given an endowment of $X \subseteq [m]$. The inflated value of items in the endowed set $X$ is captured by a gain function $g^{X}(\cdot)$, which maps every bundle $Y$ to some real value. 
Specifically, the endowed valuation, given an endowment of $X$, is assumed to take the form
$$
v^{X}(Y) = v(Y) + g^X(X \cap Y).
$$
For every buyer $i$ and every bundle $X$, $g_{i}^{X}$ denotes the gain function of $i$ with respect to endowment $X$.
Let $g_i = \{g_i^X\}_{X \subseteq [m]}$ denote the gain functions of buyer $i$, and let $g = \{g_i\}_{i \in [n]}$ denote the set of gain functions of all buyers.







An endowment equilibrium is then a Walrasian equilibrium with respect to the endowed valuations. A formal definition follows.

\begin{definition} [\textbf{Endowment equilibrium (EE) \cite{BDO18}, \cite{EFF19}}]
\label{def:ee}
A pair  $(\mathbf{S}, \mathbf{p})$ of an allocation $\allocs = (\alloci{1}, \ldots, \alloci{n})$ and item prices $\mathbf{p} = (p_1, \ldots, p_m)$, is called an endowment equilibrium with respect to $g$ if it is a WE with respect to the endowed valuations $(v_1^{S_1}, \ldots, v_n^{S_n})$, i.e:
\begin{itemize}
	\item [1.] \textbf{Utility maximization}: Every buyer receives an allocation that maximizes her utility given the item prices, i.e., $v_i^{S_i}(S_i) - \sum_{j \in S_i}p_j \geq v_i^{S_i}(T) - \sum_{j \in T}p_j$  for every $i \in [n]$ and bundle $T \subseteq [m]$.
	\item [2.] \textbf{Market clearance}: All items are allocated.
\end{itemize}
\end{definition}

\subsection{Relation between endowment equilibrium and 2PE}

In this section we show how to transform an EE into a 2PE and vice versa. 
That is, we show how high the endowment effect should be in order to compensate for a given price difference of a 2PE, and how low a price vector needs to be to make all buyers happy with their bundles in the absence of the corresponding endowment effects.

\begin{proposition}
\label{prop:2pe-ee}
Consider a valuation profile $\vals$. Given a 2PE $(\mathbf{S},\mathbf{\hat{p}},\mathbf{\check{p}})$, let $\Delta p_j = \hat{p}_j - \check{p}_j$ for every item $j \in [m]$. Then, for every endowed valuation profile $\vals^{\mathbf{S}} = (v_1^{S_1}, \ldots, v_n^{S_n})$  s.t. $g^{S_i}(S_i \backslash T \mid S_i \cap T) \ge \sum_{j \in S_i \backslash T} \Delta p_j$ for every $i \in [n]$ and $T \in [m]$, $(\mathbf{S},\mathbf{\hat{p}})$ is an EE with respect to $g$.
\end{proposition}

\begin{proof}
    Let $\vals^{\mathbf{S}} = (v_1^{S_1}, \ldots, v_n^{S_n})$ be an endowed valuation profile with corresponding gain functions $g_1^{S_1},\ldots,g_n^{S_n}$. Let $u^{2PE}_i$ denote buyer $i$'s utility under the 2PE $(\mathbf{S},\mathbf{\hat{p}},\mathbf{\check{p}})$; that is, 
    $$u^{2PE}_i(T) = v_i(T) - \sum_{j \in S_i \cap T} \check{p}_j - \sum_{j \in T \backslash S_i} \hat{p}_j.$$
    Let $u^{EE}_i$ denote buyer $i$'s utility in $(\mathbf{S},\mathbf{\hat{p}})$ under endowed valuations $\vals^{\mathbf{S}}$; i.e., 
    $$u^{EE}_i (T) = v_i(T) + g^{S_i}(T \cap S_i) - \sum_{j \in T} \hat{p}_j.$$ 
    
    
    By rearranging we get $$u^{EE}_i (S_i) = u^{2PE}_i (S_i) - \sum_{j \in S_i} \Delta p_j + g^{S_i}(S_i)$$
    and $$u^{EE}_i (T) = u^{2PE}_i (T) - \sum_{j \in S_i \cap T} \Delta p_j + g^{S_i}(S_i \cap T).$$
    Therefore, $u^{EE}_i (S_i) - u^{EE}_i (T) \ge 0$ if and only if $g^{S_i}(S_i \backslash T \mid S_i \cap T)
\ge \sum_{j \in S_i \backslash T} \Delta p_j - \left[ u^{2PE}_i (S_i) - u^{2PE}_i (T) \right]$. 
As $(\mathbf{S},\mathbf{\hat{p}},\mathbf{\check{p}})$ is a 2PE, $u^{2PE}_i (S_i) - u^{2PE}_i (T) \ge 0$ for every $T \in [m]$. Hence, if $g^{S_i}(S_i \backslash T \mid S_i \cap T)
\ge \sum_{j \in S_i \backslash T} \Delta p_j$, then $(\mathbf{S},\mathbf{\hat{p}})$ is an EE as desired.
\end{proof}

\begin{proposition}
\label{prop:ee-2pe}
Consider a valuation profile $\vals$ and gain functions $g$. Let $(\mathbf{S},\mathbf{\hat{p}})$ be an EE with respect to $g$. Then, for every low price vector $\lowprice$ s.t. 
\begin{eqnarray}
\label{eqn:ee-2pe-eq}
\sum_{j \in S_i \backslash T} \check{p}_j \le max \set{0, \sum_{j \in S_i \backslash T} \hat{p}_j - g^{S_i}(S_i \backslash T \mid S_i \cap T)}
\end{eqnarray}
for every $T \in [m]$, $(\mathbf{S},\mathbf{\hat{p}},\mathbf{\check{p}})$ is a 2PE for $\vals$.
In particular, $(\mathbf{S},\mathbf{\hat{p}},\mathbf{0})$ is a 2PE for $\vals$.
\end{proposition}

\begin{proof}
    Let  $\mathbf{\check{p}}$ be some price vector satisfying equation (\ref{eqn:ee-2pe-eq}) for every $T \in [m]$. Let $u^{2PE}_i$ and $u^{EE}_i$ be as in the proof of proposition \ref{prop:2pe-ee}. By rearrangement we get $$u^{2PE}_i (S_i) - u^{2PE}_i (T) = u^{EE}_i (S_i) - u^{EE}_i (T) + 
\sum_{j \in S_i \backslash T} \Delta p_j -
g^{S_i}(S_i \backslash T \mid S_i \cap T).$$
Therefore, $(\mathbf{S},\mathbf{\hat{p}}, \mathbf{\check{p}})$ is a 2PE if $\sum_{j \in S_i \backslash T} \check{p}_j \le \sum_{j \in S_i \backslash T} \hat{p}_j -
g^{S_i}(S_i \backslash T \mid S_i \cap T) + \left[ u^{EE}_i (S_i) - u^{EE}_i (T) \right]$.
As $(\mathbf{S},\mathbf{\hat{p}})$ is an EE w.r.t. $g$, $u^{EE}_i (S_i) - u^{EE}_i (T) \ge 0$ for every $T \in [m]$. Hence, it suffices to require that $\sum_{j \in S_i \backslash T} \check{p}_j \le \sum_{j \in S_i \backslash T} \hat{p}_j -
g^{S_i}(S_i \backslash T \mid S_i \cap T)$ for every $T \in [m]$.

We next show that in those cases where the right-hand side of the last inequality is negative, then $(\mathbf{S},\mathbf{\hat{p}}, \mathbf{0})$ is a 2PE.


To see this, note first that for every buyer $i$ and every set $T \in [m]$ $$u^{2PE}_i(T) = v_i(T) - \sum_{j \in S_i \cap T} \check{p}_j - \sum_{j \in T \backslash S_i} \hat{p}_j = v_i(T) - \sum_{j \in T \backslash S_i} \hat{p}_j \le v_i(T \cup S_i) - \sum_{j \in T \backslash S_i} \hat{p}_j = u^{2PE}_i(S_i \cup T),$$ where the inequality follows from monotonicity. Second, since $(\mathbf{S},\mathbf{\hat{p}})$ is an EE, for every buyer $i$ and every set $T \in [m]$, $v_i(S_i) + g^{S_i}(S_i) - \sum_{j \in S_i} \hat{p}_j \ge v_i(S_i \cup T) + g^{S_i}(S_i) - \sum_{j \in S_i \cup T} \hat{p}_j$, i.e.,  $v_i(S_i) \ge v_i(S_i \cup T) - \sum_{j \in T \setminus S_i} \hat{p}_j$.
Therefore, $u^{2PE}_i(S_i) \ge u^{2PE}_i(S_i \cup T)$, and we get $u^{2PE}_i(S_i) \ge u^{2PE}_i(T)$, as desired.
\end{proof}

\subsection{A Weaker Gain Function for XOS Valuations}

\label{sec:XOS_endow}

\citet{EFF19} define a partial order over gain functions: Fix a valuation function $v$ and two gain functions, $g$ and $\hat{g}$. Given an endowed set $X$, we write that $g^X \prec \hat{g}^X$ if for all $Z \subseteq X$, $g^X(Z \mid X \setminus Z) \le \hat{g}^X(Z \mid X \setminus Z)$. We write $g \prec \hat{g}$ if $g^X \prec \hat{g}^X$ for every $X \subseteq [m]$, and we say that $g$ is \emph{weaker} than $\hat{g}$.

The Identity (ID) gain function is defined as
$g^X_{ID}(Z) = v(Z)$
\cite{BDO18}.
The Absolute Loss (AL) gain function is defined as $g^X_{AL}(Z) = v(X) - v(X \setminus Z)$ \cite{EFF19}.
\citet{EFF19} show that every market with XOS buyers admits an EE with respect to the AL gain function.

In this Section we introduce a new gain function for markets with XOS buyers, which we call \emph{Supporting Prices} (SP).
We show that for such markets, the SP gain function is weaker than AL, but stronger than ID.

Let us first recall the definition of supporting prices.

\begin{definition} [\textbf{Supporting Prices \cite{DNS10}}]
\label{def:sp}
Given a valuation $v$ and a set $X \in [m]$, the prices $\{p_j\}_{j\in X}$ are supporting prices w.r.t. $v$ and $X$ if $v(X) = \sum_{j \in X} p_j$ and for every $Z \subseteq X$, $v(Z) \ge \sum_{j \in Z}p_j$.
\end{definition}

Let $v$ be an XOS valuation function, and let $\set{p^X_j}_{j \in X}$ be the supporting prices with respect to the function $v$ and the set $X \subseteq [m]$.
The SP gain function is given by $g^X_{SP}(Z) = \sum_{j \in Z} p^X_j$, for every set $Z \subseteq X$.



\begin{proposition}
For any XOS valuation, the Supporting Prices gain function is weaker than the AL gain function and stronger than the ID gain function, i.e., $g_{ID} \prec g_{SP} \prec g_{AL}$.
\end{proposition}

\begin{proof}


We first show that $g_{ID} \prec g_{SP}$.
For every $X \subseteq [m]$ and $Z \subseteq X$,
\begin{eqnarray*}
g^X_{SP}(Z \mid X \setminus Z) 
& = &
g^X_{SP}(X) - g^X_{SP}(X \setminus Z) \\
& = &
\sum_{j \in X} p^X_j - \sum_{j \in X \setminus Z} p^X_j \\
& = &
v(X) - \sum_{j \in X \setminus Z} p^X_j\\
& \ge &
v(X) - v(X \setminus Z) \\
& = &
g^X_{ID}(X) - g^X_{I}(X \setminus Z) \\
& = &
g^X_{ID}(Z \mid X \setminus Z),
\end{eqnarray*}
where the inequality follows from the definition of supporting prices.
 
We now show that $g_{SP} \prec g_{AL}$. 
For every $X \subseteq [m]$ and $Z \subseteq X$,
\begin{eqnarray*}
g^X_{SP}(Z \mid X \setminus Z) 
& = &
g^X_{SP}(X) - g^X_{SP}(X \setminus Z) \\
& = &
\sum_{j \in X} p^X_j - \sum_{j \in X \setminus Z} p^X_j \\
& = &
\sum_{j \in Z} p^X_j \\
& \le &
v(Z),
\end{eqnarray*}
where the inequality follows from the definition of supporting prices.
On the other hand,
\begin{eqnarray*}
g^X_{AL}(Z \mid X \setminus Z) 
& = &
g^X_{AL}(X) - g^X_{AL}(X \setminus Z) \\
& = &
v(X) -(v(X)-v(X \setminus (X \setminus Z))) \\
& = &
v(Z).
\end{eqnarray*}
Therefore, $g^X_{SP}(Z \mid X \setminus Z) \le g^X_{AL}(Z \mid X \setminus Z)$ for every $X \subseteq [m]$ and every $Z \subseteq X$, i.e., $g_{SP} \prec g_{AL}$.
\end{proof}

We next show that every market with XOS buyers admits an endowment equilibrium with SP gain functions.

\begin{claim}
Consider an XOS valuation profile $\vals$. Then,
\begin{enumerate}
\item
There exists a 2PE, $(\mathbf{S}, \mathbf{\hat{p}}, \mathbf{0})$, where for every $j \in S_i$, $\hat{p}_j = p^{S_i}_j$.
\item
$(\mathbf{S},\mathbf{\hat{p}})$ is an EE with SP gain functions.
\end{enumerate}
\end{claim}

\begin{proof}
\citet{CKS08} showed that in a market with XOS buyers, there always exists a S2PA PNE bid profile, where $b_{ij}$ are the supporting prices w.r.t. $v_i$ and $S_i$ if $j \in S_i$, and $0$ otherwise. 
According to Proposition \ref{prop:2pe-s2pa_pne}, the triplet $(\mathbf{S}, \mathbf{\hat{p}}, \mathbf{0})$ is a 2PE, where for every $j \in S_i$, $\hat{p}_j = p^{S_i}_j$.

By Proposition \ref{prop:2pe-ee}, if $g^{S_i}_{SP}(S_i \backslash T \mid S_i \cap T) \ge \sum_{j \in S_i \backslash T} \Delta p_j$ for every buyer $i$ and every $T \in [m]$, then $(\mathbf{S},\mathbf{\hat{p}})$ is an EE. Notice that, $g^{S_i}_{SP}(S_i \backslash T \mid S_i \cap T) = \sum_{j \in S_i \backslash T} p^{S_i}_j = \sum_{j \in S_i \backslash T} \hat{p}_j = \sum_{j \in S_i \backslash T} \Delta p_j$, where the first equality follows from the definition of the SP gain function, the second equality is given and the third equality is due to the fact that $\check{p}_j = 0$ for every $j \in [m]$. 
Therefore, $(\mathbf{S},\mathbf{\hat{p}})$ is an EE.
\end{proof}

\bibliographystyle{plainnat}
\bibliography{BibFile}

\appendix
\section{Valuation classes}
\label{sec:hetro_val}
Definitions for the valuation classes over heterogeneous items:
\begin{itemize}
	\item Unit demand: there exist $m$ values $v^1,\ldots,v^m$, s.t. $v(S) = max_{j \in S}\{v^j\}$
	\item Additive: there exist $m$ values $v^1,\ldots,v^m$, s.t. $v(S) = \sum_{j \in S}\{v^j\}$
	\item Submodular: for any $S,T \subseteq [m]$ it holds that $v(S) + v(T) \geq v(S \cup T) + v(S \cap T)$
	\item XOS: there exist vectors $v_1, \ldots, v_k \in \reals^{m}$ s.t. for any $S \subseteq [m]$ it holds that $v(S) = max_{i \in [k]}\sum_{j \in S}v_i(j)$
	\item Subadditive: for any $S,T \subseteq [m]$ it holds that $v(S) + v(T) \geq v(S\cup T)$
\end{itemize}

\section{An upper bound on the discrepancy of the optimal allocation}
\label{sec:d_upper_m}
In this section we show that for any valuation profile, every welfare-maximizing allocation can be supported in a 2PE with a discrepancy of at most $m$. The proof of Proposition \ref{prop:upper_bound} is similar to the proof of Proposition 7.1 in \cite{EFF19}.

\begin{proposition}
\label{prop:upper_bound}
Consider valuation profile $\vals$. Let $\mathbf{S}^*$ be an optimal allocation and $\hat{p}_j = v_{i(j)}(S^*_{i(j)})$ where $i(j)$ is the buyer who gets item $j$ in $\mathbf{S}^*$. Then, $(\mathbf{S}^*, \mathbf{\hat{p}}, \mathbf{0})$ is a 2PE and $D(\mathbf{S}^*) \le m$.
\end{proposition}
\begin{proof}
We first show that  $(\mathbf{S}^*, \mathbf{\hat{p}}, \mathbf{0})$ is a 2PE. Let $T \subseteq [m]$ be some bundle. Then, 

\begin{eqnarray*}
v_i(T) - \sum_{j \in T \setminus S_i^*} \hat{p}_j
& \le & 
v_i(T\cup S_i^*) - \sum_{j \in T \setminus S_i^*} \hat{p}_j \\
& = &
v_i(S_i^*) + v_i(T \setminus S_i^*|S_i^*) - \sum_{j \in T \setminus S_i^*} v_{i(j)}(S_{i(j)}^*) \\
& = &
v_i(S_i^*) + v_i(T \setminus S_i^*|S_i^*) - \sum_{i' \ne i} |T \cap S_i^*|v_{i'}(S_{i'}^*) \\
& \le &
v_i(S_i^*) + v_i(T \setminus S_i^*|S_i^*) - \sum_{i' \ne i} |T \cap S_i^*|v_{i'}(S_{i'}^* \cap T|S_{i'}^* \setminus T) j\\
& \le &
v_i(S_i^*) + v_i(T \setminus S_i^*|S_i^*) - \sum_{i' \ne i} v_{i'}(S_{i'}^* \cap T|S_{i'}^* \setminus T) \\
& \le &
v_i(S_i^*)
\end{eqnarray*}


where the first and second inequality follows by monotonicity, the third inequality follows since equality holds whenever $|T \cap S_i^*| \le 1$, and strict inequality holds otherwise. The last inequality is due to $\mathbf{S}^*$ optimality. Thus $(\mathbf{S}^*, \mathbf{\hat{p}}, \mathbf{0})$ is a 2PE.
Next, we bound the discrepancy of the optimal allocation,
$$D(\mathbf{S}^*) \leq D(\mathbf{S}^*, \mathbf{\hat{p}}, \mathbf{0}) = \frac{\sum_{i \in [n]}\sum_{j \in S^*_i} \hat{p}_j}{SW(\mathbf{S}^*, \vals)} = \frac{\sum_{i \in [n]}\sum_{j \in S^*_i} v_i(S^*_i)}{SW(\mathbf{S}^*, \vals)} \leq$$ $$\frac{m\sum_{i \in [n]}v_i(S^*_i)}{SW(\mathbf{S}^*, \vals)} = m.$$
\end{proof}

\section{Missing proofs and lemmas for section \ref{sec:vals_id}}
\label{sec:miss_valis_id}


The following lemma shows that for every submodular valuation $v$, the $(k,r)$-max-forward-slope and $(k,r)$-min-backward-slope are realized at $l'=1$, hence they do not depend on the horizon $r$.

\begin{lemma}
\label{lem:sm-mos-mbs}
For every submodular symmetric valuation
\begin{enumerate}
\item
$\overrightarrow{\Delta}_{v}(k,r) = \overrightarrow{\Delta}_{v}(k,1) = v(k+1) - v(k)$, for every $0 \le k < m$ and every $1 \le r \le m-k$.

\item
$\overleftarrow{\Delta}_{v}(k,r) = \overleftarrow{\Delta}_{v}(k,1) = v(k) - v(k-1)$, for every $0 < k \le m$ and every $1 \le r \le k$.

\end{enumerate}
\end{lemma}

\begin{proof}
For every $l \ge 1$, it holds that
$$v(k+l) - v(k) = v(1 \mid k) + v(1 \mid k+1) + \ldots + v(1 \mid k+l-1) \le l \cdot v(1 \mid k)$$
where the inequality follows from submodularity. By rearranging we get that
$$ v(k+1) - v(k) \ge \frac{v(k+l) - v(k)}{l}$$
It follows that $l=1$ realizes $\overrightarrow{\Delta}_{v}(k,r)$ for every $k$, $r$.

Similarly, for every $1 \le l \le k$, it holds that
$$v(k) - v(k-l) = v(1 \mid k-l) + v(1 \mid k-l+1) + \ldots + v(1 \mid k-1) \ge l \cdot v(1 \mid k-1)$$
where the inequality follows from submodularity. By rearranging we get that
$$v(k) - v(k-1) \le \frac{v(k)-v(k-l)}{l}$$
It follows that $l=1$ realizes  $\overleftarrow{\Delta}_{v}(k,r)$ for every $k$, $r$.
\end{proof}

The following lemma shows that the the max-forward-slope and min-backward-slope of the SM-closure always equal to the slope of the corresponding triangle.


\begin{lemma}
\label{lem:tri-sm-slope}
For every $0 \le l < |I_v|-2$

\begin{enumerate}
    \item For every $i_{l} \le k < i_{l+1}$ $\overrightarrow{\Delta}_{\tilde{v}}(k) = \alpha_l$
    \item For every $i_{l} < k \le i_{l+1}$ $\overleftarrow{\Delta}_{\tilde{v}}(k) = \alpha_l$
\end{enumerate}
\end{lemma}

\begin{proof}
By Lemma \ref{lem:sm-mos-mbs}, 
for every $0 \le k < m$,
$\overrightarrow{\Delta}_{\tilde{v}}(k) = \tilde{v}(k+1) - \tilde{v}(k)$.
Since the valuation $\tilde{v}$ is linear in the range $(i_{l}, i_{l+1})$, the max-forward-slope does not change inside the triangle range, and it equals to:
$$\overrightarrow{\Delta}_{\tilde{v}}(k) = \frac{\tilde{v}(i_{l+1}) - \tilde{v}(i_{l})}{i_{l+1} - i_{l}} = \frac{v(i_{l+1}) - v(i_{l})}{i_{l+1} - i_{l}} = \alpha_l$$
for every $i_{l} \le k < i_{l+1}$.

Similarly, by Lemma \ref{lem:sm-mos-mbs}, for every $0 < k \le m$ $\overleftarrow{\Delta}_{\tilde{v}}(k) = \tilde{v}(k) - \tilde{v}(k-1)$. Since the valuation $\tilde{v}$ is linear in the range $(i_{l}, i_{l+1})$, the min-backward-slope does not change inside the triangle range, and it equals to:
$$\overleftarrow{\Delta}_{\tilde{v}}(k) = \frac{\tilde{v}(i_{l+1}) - \tilde{v}(i_{l})}{i_{l+1} - i_{l}} = \frac{v(i_{l+1}) - v(i_{l})}{i_{l+1} - i_{l}} = \alpha_l$$
for every $i_{l} < k \le i_{l+1}$.
\end{proof}

The following lemma shows that the max-forward-slope of the valuation $v$ at $k = i_{l}$ equals to the slope of the corresponding triangle, and that the max-forward-slope of any $k \in T_{{l}}$ is realized at $l' \le i_{l+1}-k$.

\begin{lemma}
\label{lem:tri-v-slope}
For every triangle $T_{{l}}$, 

\begin{enumerate}
\item
\label{lem:tri-v-slope-1}
$\overrightarrow{\Delta}_{v}(i_{l},i_{l+1}-i_{l}) = \alpha_l$.


\item
\label{lem:tri-v-slope-2}
$\overrightarrow{\Delta}_{v}(i_{l})$ is realized at $l' = i_{l+1}-i_{l}$. 



\item
\label{lem:tri-v-slope-3}
For every $i_{l} \le k < i_{l+1}$, $\overrightarrow{\Delta}_{v}(k)$ is realized at $l' \le i_{l+1}-k$.
\end{enumerate}

\end{lemma}

\begin{proof}
As $i_{l+1}$ is the first point after $i_{l}$ in which $v(k) = \tilde{v}(k)$, we have that for every point $i_{l} < k < i_{l+1}$, $v(k) < \tilde{v}(k)$. Moreover, since $\tilde{v}$ is linear in the triangle, the $(i_l, i_{l+1}-i_{l})$-max-forward-slope of $v$, is realized at $l' = i_{l+1}-i_{l}$.
That is, 
$$\overrightarrow{\Delta}_{v}(i_{l},i_{l+1}-i_{l}) = \max_{l = 1, 2, \ldots, i_{l+1}-i_{l}} \set{\frac{v(i_{l}+l)-v(i_{l})}{l}} = \frac{v(i_{l+1})-v(i_{l})}{i_{l+1}-i_{l}} = \alpha_l$$
where the last equality follows from the definition of $\alpha_l$.

As submodular functions are convex, the slope of the triangles is non decreasing, and therefore the slope of all the following triangles after $T_{{l}}$ are at most $\alpha_l$. Since $v(k) \le \tilde{v}(k)$ for every $0 \le k \le m$, the $(k,r)$-max-forward-slope of $v$ at point $k=i_{l}$ cannot increase beyond $\alpha_l$ for $r \ge i_{l+1}-i_{l}$, i.e., $\overrightarrow{\Delta}_{v}(i_{l}) = \overrightarrow{\Delta}_{v}(i_{l},i_{l+1}-i_{l})$.
For the same argument, we conclude that $\overrightarrow{\Delta}_{v}(k) = \overrightarrow{\Delta}_{v}(k,i_{l+1}-k)$, for every $i_{l} \le k < i_{l+1}$.
That is, $\overrightarrow{\Delta}_{v}(k)$ is realized at $l' \le i_{l+1}-k$ for $k \in T_{{l}}$, and at $l' = i_{l+1}-i_{l}$ for $k = i_{l}$.
\end{proof}

The following lemma shows that the min-backward-slope of the valuation $v$ at $k = i_{l+1}$ equals to the slope of triangle $T_{l}$, and that the min-backward-slope of any $i_l < k \le i_{l+1}$ is realized at $l' \le k - i_{l}$.

\begin{lemma}
\label{lem:tri-v-slope_back}
For every triangle $T_{{l}}$, 

\begin{enumerate}
\item
\label{lem:tri-v-slope-1_back}
$\overleftarrow{\Delta}_{v}(i_{l+1},i_{l+1}-i_{l}) = \alpha_l$.

\item
\label{lem:tri-v-slope-2_back}
$\overleftarrow{\Delta}_{v}(i_{l+1})$ is realized at $l' = i_{l+1}-i_{l}$. 

\item
\label{lem:tri-v-slope-3_back}
For every $i_{l} < k \le i_{l+1}$, $\overleftarrow{\Delta}_{v}(k)$ is realized at $l' \le k - i_{l}$.
\end{enumerate}

\end{lemma}

\begin{proof}
Since $\tilde{v}$ is the SM-Closure of $v$, for every $i_{l} < k < i_{l+1}$, $v(k) < \tilde{v}(k)$. Since $\tilde{v}$ is linear in $T_{l}$ we have that the $(i_{l+1}, i_{l+1} - i_l)$-min-backward-slope is realized at $l' = i_{l+1} - i_l$. That is:
$$\overleftarrow{\Delta}_{v}(i_{l+1},i_{l+1}-i_{l}) = \min_{l = 1, 2, \ldots, i_{l+1}-i_{l}} \set{\frac{v(i_{l+1})-v(i_{l+1} - l)}{l}} = \frac{v(i_{l+1})-v(i_l)}{i_{l+1}-i_{l}} = \alpha_l$$
where the last equality follows from the definition of $\alpha_l$.

As submodular functions are convex, the slope of the triangles is non decreasing, and therefore the slope of all the triangle that come before $T_{l}$ are at least $\alpha_l$. Thus, the $(i_{l+1},r)$-min-backward-slope cannot decrease beyond $\alpha_l$ for $r \ge i_{l+1}-i_{l}$, i.e., $\overleftarrow{\Delta}_{v}(i_{l+1}) = \overleftarrow{\Delta}_{v}(i_{l+1},i_{l+1}-i_{l})$.
For the same argument, we conclude that $\overleftarrow{\Delta}_{v}(k) = \overleftarrow{\Delta}_{v}(k,k - i_{l})$, for every $i_{l} < k \le i_{l+1}$.
That is, $\overleftarrow{\Delta}_{v}(k)$ is realized at $l' \le k - i_{l}$ for $i_l < k \le i_{l+1}$, and at $l' = i_{l+1}-i_{l}$ for $k = i_{l+1}$.
\end{proof}

The following lemma shows that at the intersecting indices, both the max-forward-slope and min-backward-slope of $v$ and $\tilde{v}$, equal to the slope of the corresponding triangle.

\begin{lemma}
\label{lem:tri-slope}
For every $i_{l} \in I_v$, 
\begin{enumerate}
    \item $\overrightarrow{\Delta}_{v}(i_{l}) = \overrightarrow{\Delta}_{\tilde{v}}(i_{l}) = \alpha_l$.
    \item $\overleftarrow{\Delta}_{v}(i_{l+1}) = \overleftarrow{\Delta}_{\tilde{v}}(i_{l+1}) = \alpha_l$.
\end{enumerate}
\end{lemma}

\begin{proof}
The proof follows directly from Lemmas \ref{lem:tri-sm-slope}, \ref{lem:tri-v-slope} and \ref{lem:tri-v-slope_back}.
\end{proof}

The following lemma shows that the slopes of the triangles are monotonically decreasing.

\begin{lemma}
\label{lem:mon-tri-slopes}
$\alpha_0 \ge \alpha_1 \ge \ldots \ge \alpha_{|I_v|-2}$.
\end{lemma}

\begin{proof}
The proof follows directly from Lemma \ref{lem:sm-mos-mbs}, \ref{lem:tri-sm-slope} and the fact that $\tilde{v}$ is submodular.
\end{proof}

We next show that the max-forward-slope (respectively, min-backward-slope) of the left-most (resp., right-most) point of a triangle is the minimal max-forward-slope (resp., maximal min-backward-slope) among all triangle points.


\begin{claim}
\label{clm:min-slope-in-tri}
For every triangle $T_{{l}}$, for every $i_{l} < k < i_{l+1}$,
\begin{enumerate}
    \item $\overrightarrow{\Delta}_{v}(i_{l}) < \overrightarrow{\Delta}_{v}(k)$
    \item $\overleftarrow{\Delta}_{v}(k) < \overleftarrow{\Delta}_{v}(i_{l+1})$
\end{enumerate}

\end{claim}

\begin{proof}
Since $i_{l}$ and $i_{l+1}$ are two consecutive intersecting indices, for every point $i_{l} < k < i_{l+1}$, $v(k) < \tilde{v}(k)$. Therefore, 
$$\overrightarrow{\Delta}_{v}(i_{l}) = \overrightarrow{\Delta}_{\tilde{v}}(i_{l}) = \overrightarrow{\Delta}_{\tilde{v}}(k) = \frac{\tilde{v}(i_{l+1})-\tilde{v}(k)}{i_{l+1}-k} < \frac{v(i_{l+1})-v(k)}{i_{l+1}-k} \le \overrightarrow{\Delta}_{v}(k,i_{l+1}-k) = \overrightarrow{\Delta}_{v}(k),$$
where the first equality is due to Lemma \ref{lem:tri-slope}, the second equality is due to Lemma \ref{lem:tri-sm-slope}, the last inequality follows from Definition \ref{def:fdelta} and the last equality follows from Lemma \ref{lem:tri-v-slope}.

In addition to that we have, 
$$\overleftarrow{\Delta}_{v}(i_{l+1}) = \overleftarrow{\Delta}_{\tilde{v}}(i_{l+1}) = \overleftarrow{\Delta}_{\tilde{v}}(k) = \frac{\tilde{v}(k)-\tilde{v}(i_{l})}{k-i_{l}} > \frac{v(k)-v(i_{l})}{k-i_{l}} \ge \overleftarrow{\Delta}_{v}(k,k-i_{l}) = \overleftarrow{\Delta}_{v}(k),$$
where the first equality is due to Lemma \ref{lem:tri-slope}, the second equality is due to Lemma \ref{lem:tri-sm-slope}, the last inequality follows from Definition \ref{def:bdelta} and the last equality follows from Lemma \ref{lem:tri-v-slope_back}.
\end{proof}

From Claim \ref{clm:min-slope-in-tri}, Lemma \ref{lem:tri-slope} and Lemma \ref{lem:mon-tri-slopes}, we conclude that the minimum value of the max-forward-slope up to a point $k$, equals to the slope of the last triangle that precedes $k$, i.e.,
\begin{corollary}
\label{cor:min-slope}
For every $0 < k \le m$, let $i_j \in I_v$ be the largest intersecting index that is smaller than $k$. It holds that $\min_{0 \le k' < k} \set{\overrightarrow{\Delta}_{v}(k')} = \alpha_j$.
\end{corollary}

\begin{proof} [Proof of Lemma \ref{lem:intgr-lem}]
Consider the SM-closure, $\tilde{v}$.  
Notice that the valuation at any intersecting index $i_l \in I_v$ equals to the sum of previous triangle increases, i.e. $v(i_l) = \sum_{j=0}^{l-1} (i_{j+1} - i_{j}) \cdot \alpha_j$.

Let $i_{l'} \in I_v$ be the largest intersecting index that is less than $k$.
If $k = i_{l'+1}$, 
then $v(k) = v(i_{l'+1}) = \sum_{j=0}^{l'} (i_{j+1} - i_{j}) \cdot \alpha_j \ge \alpha_{l'} \cdot \sum_{j=0}^{l'} (i_{j+1} - i_{j}) = \alpha_{l'} \cdot i_{l'+1} = \min_{0 \le k' < k} \set{\overrightarrow{\Delta}_{v}(k')} \cdot k$, where the first inequality follows from Lemma \ref{lem:mon-tri-slopes}, the third equality is due to telescoping sum, and the last equality follows from Corollary \ref{cor:min-slope}.

If $k \neq i_{l'+1}$, then $i_{l'} < k < i_{l'+1}$ and $k \in T_{{l'}}$. 
Let $i_{l''} \in I_v$ be the largest intersecting index that is less than $i_{l'+1} - k$. Notice that $i_{l''} \le i_{l'}$.
Also, 
\begin{eqnarray}
\nonumber
v(i_{l'+1} - k) 
& \le & 
v(i_{l''}) + ((i_{l'+1} - k) - i_{l''}) \cdot \alpha_{l''} \\
\label{eqn:intgr-lem-1}
& = &
\sum_{j=0}^{l''-1} (i_{j+1} - i_{j}) \cdot \alpha_j  + ((i_{l'+1} - k) - i_{l''}) \cdot \alpha_{l''}
\end{eqnarray}
and,
\begin{eqnarray}
\nonumber
v(i_{l'+1}) 
& = &
\sum_{j=0}^{l'} (i_{j+1} - i_{j}) \cdot \alpha_j \\
\label{eqn:intgr-lem-2}
& = &
\sum_{j=0}^{l''-1} (i_{j+1} - i_{j}) \cdot \alpha_j + (i_{l''+1} - i_{l''}) \cdot \alpha_{l''} + \sum_{j=l''+1}^{l'} (i_{j+1} - i_{j}) \cdot \alpha_j
\end{eqnarray}
Using subadditivity together with Equations (\ref{eqn:intgr-lem-1}) and (\ref{eqn:intgr-lem-2}) we get, 
\begin{eqnarray*}
v(k) 
& \ge &
v(i_{l'+1}) - v(i_{l'+1} - k) \\
& \ge &
((i_{l''+1} - i_{l''}) - ((i_{l'+1} - k) - i_{l''})) \cdot \alpha_{l''} + \sum_{j=l''+1}^{l'} (i_{j+1} - i_{j}) \cdot \alpha_j \\
& = &
(i_{l''+1} - (i_{l'+1} - k)) \cdot \alpha_{l''} + \sum_{j=l''+1}^{l'} (i_{j+1} - i_{j}) \cdot \alpha_j \\
& \ge &
(i_{l''+1} - (i_{l'+1} - k)) \cdot \alpha_{l'}  + \sum_{j=l''+1}^{l'} (i_{j+1} - i_{j}) \cdot \alpha_{l'} \\
& = &
\left( 
i_{l''+1} - (i_{l'+1} - k) + \sum_{j=l''+1}^{l'} (i_{j+1} - i_{j})
\right) 
\cdot \alpha_{l'} \\
& = & 
k \cdot \alpha_{l'} \\
& = &
k \cdot \min_{0 \le k' < k} \set{\overrightarrow{\Delta}_{v}(k')}
\end{eqnarray*}
where the last Inequality follows from Lemma \ref{lem:mon-tri-slopes} and the last equality follows from Corollary \ref{cor:min-slope}.
\end{proof}


\begin{figure}[h!]
\begin{center}
\includegraphics[width=16cm,height=4cm]{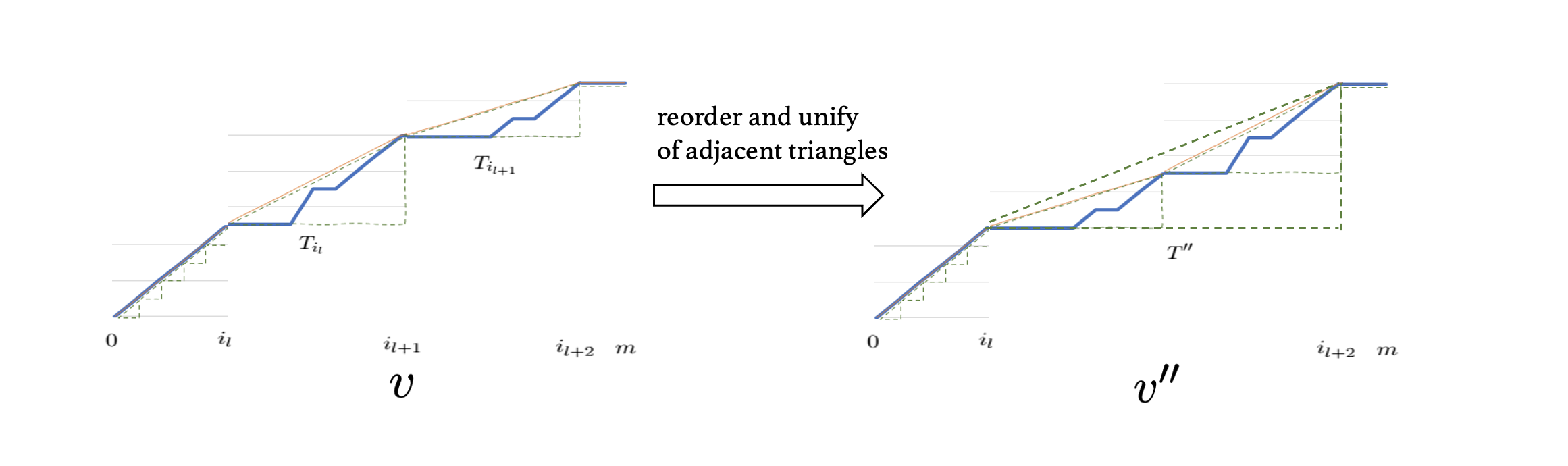}
\end{center}
\caption{Reorder and union of adjacent triangles: $v$ has 9 triangles.
$v''$ is composed of the same triangles, but with triangles $T_{{l}}$ and $T_{{l+1}}$ in the opposite order. Hence, these two triangles are unified to one larger triangle, $T''$.}
\label{fig:tri-switch}
\end{figure}

\begin{lemma}
\label{lem:sort2tri}
Let $v: [m] \rightarrow R^+$ be some valuation and $\tilde{v}$ its SM-closure.
Let $v'':[m] \rightarrow R^+$ be a valuation
s.t. $v''(k) = v(k)$ for every $0 \le k < i_{l}$ and every $i_{l+2} \le k \le m$, while in the range $i_{l} \le k < i_{l+2}$, $v''(k)$ is composed of the triangles $T_{{l}}$ and $T_{{l+1}}$ of the function $v$ in the opposite order (see Figure \ref{fig:tri-switch}). Then,
\begin{enumerate}

\item
\label{lem:sort2tri-1}
$v''(i_{l}) = v(i_{l})$ and $v''(i_{l+2}) = v(i_{l+2})$.

\item
\label{lem:sort2tri-2}
$v''$ is monotone.


\item
\label{lem:sort2tri-3}
$v''$ has only one right triangle, $T''$, in the range $i_{l} \le k \le i_{l+2}$, with slope $\alpha''$, where $\alpha_{l+1}\le \alpha'' \le \alpha_{l}$.

\item
\label{lem:sort2tri-4}
$\overrightarrow{\Delta}_{v''}(i_{l})$ is realized at $i_{l+2}$.

\item
\label{lem:sort2tri-5}
For every $0 \le k < m$, $\overrightarrow{\Delta}_{v}^s(k) \le \overrightarrow{\Delta}_{v''}^s(k)$.

\end{enumerate}
\end{lemma}

\begin{proof}
It is easy to see that $v''(i_{l}) = v(i_{l})$ and $v''(i_{l+2}) = v(i_{l+2})$, by the definition on $v''$.
Moreover, as $v$ is monotone, every \emph{slice} of $v$ is also monotone. Therefore, changing the order of some slices of $v$ keeps the monotonicity property, which leads to the conclusion that $v''$ is monotone.

Consider the range $k \in [i_{l}, i_{l+2}]$.
From Lemma \ref{lem:tri-sm-slope}, $\alpha_l = \overrightarrow{\Delta}_{\tilde{v}}(i_{l}) \ge \overrightarrow{\Delta}_{\tilde{v}}(i_{l+1}) = \alpha_{l+1}$, 
where the inequality follows from Lemma \ref{lem:sm-mos-mbs} and the fact that $\tilde{v}$ is submodular.
As in this range $v''$ is composed of the same triangles but in the opposite order, $v''$ intersect with its SM-closure in only two points: the first point, $i_{l}$,  and the last point, $i_{l+2}$. Therefore, $v''$ has only one right triangle, $T''$, in the range $k \in [i_{l}, i_{l+2}]$ and it is easy to see that the slope $\alpha''$ satisfies $\alpha_{l+1}\le \alpha'' \le \alpha_l$.
Moreover, according to Lemma \ref{lem:tri-v-slope} (\ref{lem:tri-v-slope-2}), $\overrightarrow{\Delta}_{v''}(i_{l})$ is realized at $i_{l+2}$.


Now lets look at triangle $T_{{l}}$. 
According to Lemma \ref{lem:tri-v-slope} (\ref{lem:tri-v-slope-3}), the max-forward-slope of every $k \in T_{{l}}$ in $v$ is realized at $l' \le i_{l+1} - i_{l}$ and by definition is determined only by the triangle points that comes after it. In $v''$, $T_{{l}}$ is part of a larger triangle, $T''$, but as it is located at the end of $T''$, there is no change in the max-forward-slope of $T_{{l}}$'s points when moving from $v$ to $v''$.
The above is not true for triangle $T_{{l+1}}$, which is located at the beginning of $T''$ in $v''$. Hence, the max-forward-slope of every $k \in T_{{l+1}}$ in $v''$ is influenced by the points in $T''$ that come after it, including the points of $T_{{l}}$. However, the max-forward-slope of every $k \in T_{{l+1}}$ can only increase when moving from $v$ to $v''$, as the max-forward-slope is a maximum over a larger set.
Moreover, according to Lemma \ref{lem:tri-v-slope}, the slope of all points $0 \le k < i_{l}$ and $i_{l+2} \le k \le m$ stays the same.

Therefore, $\overrightarrow{\Delta}_{v}^s(k) \le \overrightarrow{\Delta}_{v''}^s(k)$ for every $0 \le k < m$.
\end{proof}

\begin{proof} [Proof of Theorem \ref{thm:slv-le-lf}]
We prove this theorem by induction. We will denote $v_i$ as the valuation $v_i:[i] \rightarrow R^+$ and $f_i$ as its flat function.
The base case is $m = 1$, where $v_1$ is some valuation.
In this case, $v_1(k) = f_1(k)$ for every $0 \le k \le m = 1$. Therefore, $\overrightarrow{\Delta}_{v_1}^s(k) \le \overrightarrow{\Delta}_{f_1}(k)$ for every $0 \le k < m = 1$.

Assume that $\overrightarrow{\Delta}_{v_m}^s(k) \le \overrightarrow{\Delta}_{f_m}(k)$ for every $0 \le k < m$ for any valuation $v_m$. We need to show that $\overrightarrow{\Delta}_{v_{m+1}}^s(k) \le \overrightarrow{\Delta}_{f_{m+1}}(k)$ for every $0 \le k < m+1$ for any valuation $v_{m+1}$.

Consider a valuation $v_{m+1}$, it's SM-closure $\tilde{v}_{m+1}$.
Let $i_{l}, i_{l+1}$  be two consecutive intersecting indices in $I_{v_{m+1}}$ and let 
${T_{{l}}},T_{{l+1}} \in T_{v_{m+1}}$ be their corresponding two adjacent triangles. 

By Lemma \ref{lem:tri-sm-slope}, $\alpha_l = \overrightarrow{\Delta}_{\tilde{v}_{m+1}}(i_{l}) \ge \overrightarrow{\Delta}_{\tilde{v}_{m+1}}(i_{l+1}) = \alpha_{l+1}$, 
where the inequality follows from Lemma \ref{lem:sm-mos-mbs} and the fact that $\tilde{v}_{m+1}$ is submodular. That is, the list of triangle slopes is non-increasing, $\alpha_0 \ge \alpha_1 \ge \ldots \ge \alpha_{|I_{v_{m+1}}|-2}$.
For ease of presentation, we refer to the last intersecting index, which is $i_{|I_{v_{m+1}}|-1}$ as $i_{-1}$, and to the $d^{th}$ intersecting index from the end, i.e., $i_{|I_{v_{m+1}}|-d}$, as $i_{-d}$.

We now want to create a new function $v^1_{m+1}$ from the function $v_{m+1}$ by \emph{reorder and unify} of the last two triangles $T_{{-3}},T_{{-2}} \in T_{v_{m+1}}$ as described in Lemma \ref{lem:sort2tri}. There is no change in the other points $k \notin T_{{-3}},T_{{-2}}$, i.e. $v^1_{m+1}(k) = v_{m+1}(k)$ for every $i_{0} \le k < i_{-3}$.
Let $T_{{-3},{-2}}$ be the new triangle created after the reorder of unify of triangles $T_{{-3}}$ and $T_{{-2}}$, with slope $\alpha_{{-3},{-2}}$.  
According to Lemma \ref{lem:sort2tri},
$v^1_{m+1}$ is monotone,
$v^1_{m+1}(m+1) = v_{m+1}(m+1)$, 
$v^1_{m+1}(i_{-3}) = v_{m+1}(i_{-3})$,
$\alpha_{-2} \le \alpha_{{-3},{-2}} \le \alpha_{-3}$ and
$\overrightarrow{\Delta}_{v_{m+1}}^s(k) \le \overrightarrow{\Delta}_{v_{m+1}^1}^s(k)$ for every $0 \le k < m+1$.
Note that $v^1_{m+1}$ has $|T_{v_{m+1}}|-1$ triangles, i.e. one triangle less than $v_{m+1}$.
We can now create a new function $v^2_{m+1}:[m+1] \rightarrow R^+$ from the function $v^1_{m+1}$ by reorder and unify of the last two triangles of $v^1_{m+1}$. 
Similar to the previous arguments, $v^2_{m+1}$ is monotone and has $| T_{v_{m+1}}|-2$ triangles, i.e. one triangle less than $v^1_{m+1}$ and $\overrightarrow{\Delta}_{v_{m+1}^1}^s(k) \le \overrightarrow{\Delta}_{v_{m+1}^2}^s(k)$ for every $0 \le k < m+1$. We can repeat this process till the last function, $v^t_{m+1}$, has only one triangle (see Figure \ref{fig:concave-v}).

\begin{figure}[h!]
\begin{center}
\includegraphics[width=16cm,height=4cm]{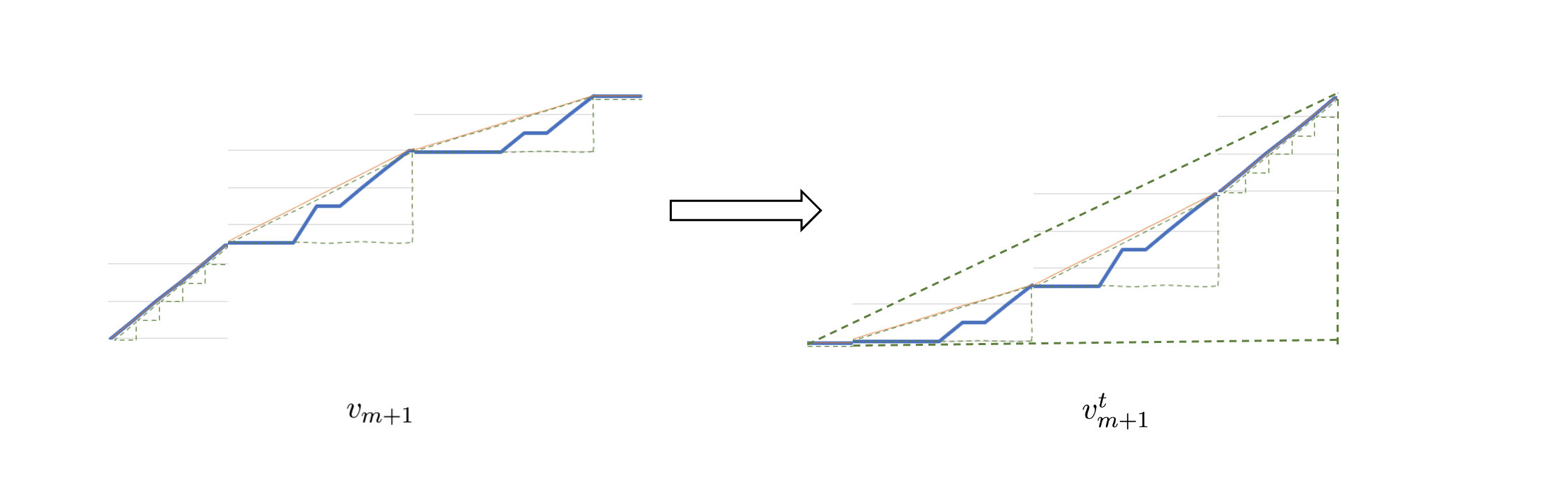}
\end{center}
\caption{The function $v_{m+1}$ with $|T_{v_{m+1}}|-1$ right triangles and the function $v^t_{m+1}$ with only one right triangle, after $|T_{v_{m+1}}|-2$ operations of \emph{reorder and union}.}
\label{fig:concave-v}
\end{figure}

Using Lemma \ref{lem:sort2tri} (\ref{lem:sort2tri-1}), we get that $v_{m+1}(0) = v^1_{m+1}(0) = v^2_{m+1}(0) = \ldots = v^t_{m+1}(0) = 0$ and $v_{m+1}(m+1) = v^1_{m+1}(m+1) = v^2_{m+1}(m+1) = \ldots = v^t_{m+1}(m+1)$, and therefore the flat function, $f_{m+1}$, which corresponds to the function $v_{m+1}$, and the flat function, $f_{m+1}^t$, which corresponds to the function $v^t_{m+1}$ are the same, i.e., $f_{m+1} = f_{m+1}^t$. Moreover, since $\overrightarrow{\Delta}_{v_{m+1}^1}^s(k) \le \overrightarrow{\Delta}_{v_{m+1}^2}^s(k) \le \ldots \le \overrightarrow{\Delta}_{v_{m+1}^t}^s(k)$ for every $0 \le k < m+1$, it is sufficient to show is that $\overrightarrow{\Delta}_{v_{m+1}^t}^s(k) \le \overrightarrow{\Delta}_{f_{m+1}^t}(k)$ for every $0 \le k < m+1$.

First notice that according to Lemma \ref{lem:sort2tri} (\ref{lem:sort2tri-4}), $\overrightarrow{\Delta}_{v^t_{m+1}}(0)$ is realized at point $m+1$.
Therefore, 
\begin{eqnarray}
\label{eqn:slv-le-lf-1}
\overrightarrow{\Delta}_{v^t_{m+1}}(0) = \frac{v^t_{m+1}(m+1)-v^t_{m+1}(0)}{m+1} = \frac{v^t_{m+1}(m+1)}{m+1} = \overrightarrow{\Delta}_{f^t_{m+1}}(0)
\end{eqnarray}
where the last equality follows from Observation \ref{obs:flat-LMFS}.

Let $v^t_m: [m] \rightarrow R^+$ be the following monotone set function: 
$$v^t_m(k) = 
	\begin{cases}
		0               & k = 0 \\
		v^t_{m+1}(k+1)   & 1 \le k \le m.
	\end{cases}$$
and let $f^t_m$ be its corresponding flat function. By the definition of $v^t_m$, we have that for every $2 \le k \le m$, 
\begin{eqnarray}
\label{eqn:slv-le-lf-2}
\overrightarrow{\Delta}_{v^t_{m+1}}(k) = \overrightarrow{\Delta}_{v^t_{m}}(k-1)
\end{eqnarray}
and for every $1 \le k \le m$, 
\begin{eqnarray}
\label{eqn:slv-le-lf-3}
\overrightarrow{\Delta}_{f^t_{m+1}}(k) = \overrightarrow{\Delta}_{f^t_{m}}(k-1)
\end{eqnarray}

As $v^t_{m+1}(k) = v^t_m(k-1)$ for every $2 \le k \le m+1$ and $v^t_{m+1}(1) \ge 0 = v^t_m(0)$, we have that,
\begin{eqnarray}
\label{eqn:slv-le-lf-4}
\overrightarrow{\Delta}_{v^t_{m+1}}(1) \le \overrightarrow{\Delta}_{v^t_{m}}(0)
\end{eqnarray}

By the induction assumption, we know that $\overrightarrow{\Delta}_{v_{m}^t}^s(k) \le \overrightarrow{\Delta}_{f_{m}^t}(k)$ for every $0 \le k < m$, and together with Equations (\ref{eqn:slv-le-lf-1}), (\ref{eqn:slv-le-lf-2}), (\ref{eqn:slv-le-lf-3}) and (\ref{eqn:slv-le-lf-4}), we get that $\overrightarrow{\Delta}_{v_{m+1}^t}^s(k) \le \overrightarrow{\Delta}_{f_{m+1}^t}(k)$ for every $0 \le k < m+1$.

\end{proof}

\section{Missing proofs for section \ref{sec:2pe_id}}
\label{miss_proof_2pe_id}

\begin{proof} [Proof of Proposition \ref{prop:u2pe_necess_cond}]

Let $k_i = |S_i|$, $T \subseteq [m]$, $t = |T|$, $t_i = |T \cap S_i|$. 
For identical items, Inequality (\ref{eqn:2pe}) can be written as 
\begin{eqnarray*}
v_i(k_i) - (k_i-t_i) \cdot \check{p}^{(i)} 
& \ge &
v_i(t) - \sum_{i' \neq i} t_{i'} \cdot \hat{p}^{(i')}
\end{eqnarray*}
for every $i \in [n]$ and every $t, t_1, t_2, \ldots, t_n$, such that $0 \le t_i \le k_i$ and $t = \sum_{i \in [n]} t_i$.
Alternatively,  
\begin{eqnarray}
\label{eqn:um_u2pe}
v_i(k_i) - v_i(t) 
& \ge &
(k_i-t_i) \cdot \check{p}^{(i)}  - \sum_{i' \neq i} t_{i'} \cdot \hat{p}^{(i')} 
\end{eqnarray}

Necessity: Assume that Inequality (\ref{eqn:um_u2pe}) holds.
Substituting $t = k_i$ and $t_i = k_i-1$ in Inequality (\ref{eqn:um_u2pe}), we get $\check{p}^{(i)} \le \hat{p}^{(i')}$ for every $i' \neq i$. Together with the assumption that $\check{p_j} \le \hat{p_j}$ for every item $j$, 
condition (\ref{um_cond_1}) follow.
Substituting 
$t_i = t \le k_i$ and $t_{i'} = 0$ for every $i' \neq i$ in Inequality (\ref{eqn:um_u2pe}), we get
$$\check{p}^{(i)} 
 \le  
\frac{v_i(k_i) - v_i(t)}{(k_i-t)},~~~ \textrm{for every $i \in [n]$ and every $0 \le t < k_i$}$$
which is equivalent to
$$\check{p}^{(i)} \le \min_{t = 0, 1, \ldots k_i-1} \set{ \frac{v_i(k_i) - v_i(t)}{(k_i-t)}} =
 \min_{l = 1, 2, \ldots k_i} \set{ \frac{v_i(k_i) - v_i(k_i-l)}{l}} = \overleftarrow{\Delta}_{v_i}(k_i) = \overleftarrow{\Delta}_{v_i}(|S_i|)$$
 for every $i \in [n]$, which is precisely condition (\ref{um_cond_2}. Condition (\ref{um_cond_3}) is obtained by applying $t_i = k_i < t$ 
 in Inequality (\ref{eqn:um_u2pe}). That is,
 \begin{eqnarray}
 \label{eqn:um_u2pe_1}
  \sum_{i' \neq i} t_{i'} \cdot \hat{p}^{(i')} \ge v_i(t) - v_i(k_i)
 \end{eqnarray}
 for every $i \in [n]$ and every $k_i < t \le m$.

Sufficiency: Assume that conditions (\ref{um_cond_1}), (\ref{um_cond_2}) and (\ref{um_cond_3}) of the proposition hold.
Fix buyer $i$, and some $t$. One can easily verify that, since $\check{p}^{(i)} \le \hat{p}^{(i')}$ for every $i' \neq i$, the right hand side of Inequality (\ref{eqn:um_u2pe}) attains its highest value when $t_i = min \set{k_i,t}$, and $\sum_{i' \neq i} t_{i'} = t - t_i$.  
Therefore, it suffices to satisfy Inequality (\ref{eqn:um_u2pe}) for this case. If $t \le k_i$, then $t_i = t$, $t_{i'} = 0$ for every $i' \neq i$, and if $t > k_i$ then $t_i = k_i$, $\sum_{i' \neq i} t_{i'} = t - k_i$. Therefore, it remains to show that: 

\begin{eqnarray}
v_i(k_i) - v_i(t) 
& \ge &
	\begin{cases}
		(k_i-t) \cdot \check{p}^{(i)} & 0 \le t \le k_i \\
		- \sum_{i' \neq i} t_{i'} \cdot \hat{p}^{(i')} &k_i < t \le m
	\end{cases} \\
\end{eqnarray}
for every $i \in [n]$ and every $t, t_1, t_2, \ldots, t_n$, s.t. $t_i = min \set{k_i,t}$, $t = \sum_{i' \in [n]} t_{i'}$ and for every $i' \in [n]$, $0 \le t_{i'} \le k_{i'}$.
One can easily verify that the above Inequality holds whenever conditions (\ref{um_cond_2}) and (\ref{um_cond_3}) of the proposition hold.



\end{proof}

\begin{proof} [Proof of Proposition \ref{prop:2pe_ecp_n_suff_cond}]
From condition (\ref{u2pe_cond_5}) we get market clearance.
From condition (\ref{u2pe_cond_4}), for every $i$ and $i' \neq i$, we have: 
$$\hat{p}^{(i')} \ge \overrightarrow{\Delta}_{v_i}(|S_{i}|) = \overrightarrow{\Delta}_{v_i}(k_i) \ge \frac{v_i(t) - v_i(k_i)}{t-k_i}$$
for every $k_i < t \le m-k_i$,
and therefore, 
$$\sum_{i' \neq i} t_{i'} \cdot \hat{p}^{(i')} \ge \sum_{i' \neq i} t_{i'} \cdot \frac{v_i(t) - v_i(k_i)}{t-k_i} = \frac{v_i(t) - v_i(k_i)}{t-k_i} \cdot \sum_{i' \neq i} t_{i'} = v_i(t) - v_i(k_i)$$
for every $t$ and $i$. That is, condition (\ref{um_cond_3}) of Proposition \ref{prop:u2pe_necess_cond} holds and together with conditions (\ref{um_cond_1}), (\ref{um_cond_2}) and (\ref{um_cond_3}) of the current proposition, all the conditions of Proposition \ref{prop:u2pe_necess_cond} hold. Therefore, utility maximization holds and hence $(\mathbf{S}, \mathbf{\hat{p}}, \mathbf{\check{p}})$ is a U-2PE.
\end{proof}

\section{A simpler lower bound}
\label{lb_simple}

\begin{theorem}
\label{thm:simple_ex}
There exists a market with identical items and 2 subadditive buyers that admits no 2PE with discrepancy smaller than $\frac{6}{5}$.
\end{theorem}

\begin{proof}
Consider a setting with two buyers with the following identical subadditive valuation over 27 identical items:
\begin{eqnarray*}
v_1(k) = v_2(k) = 
	\begin{cases}
		1 & 1 \le k \le 4 \\
		2 & 5 \le k \le 10 \\
		3 & 11 \le k \le 16 \\
		4 & 17 \le k \le 22 \\
		5 & 23 \le k \le 26 \\
		6 & k = 27 \\
	\end{cases} \\
\end{eqnarray*}

To find the minimum possible discrepancy of any 2PE for this setting, we need to consider all 14 possible allocations (up to symmetry). 
Table \ref{table:d2} gives the values of the max-forward-slope and min-backward-slope of each allocation for each buyer. One can now use Proposition \ref{prop:2pe_ecp_n_suff_cond} to  calculate the maximum $\mathbf{\check{p}}$ and minimum $\mathbf{\hat{p}}$ and the corresponding minimum discrepancy for every allocation. 
One can verify that the minimum discrepancy over all allocations is achieved for $(S_1,S_2) = (6,21)$, with discrepancy  $\frac{6}{5}$.

\begin{table}[h!]
\centering
\begin{tabular} { c | c | c | c | c | c | c | c | c | c | c  }

  $k_1$ & $k_2$ & $\overrightarrow{\Delta}_{v}(k_1)$ & $\overleftarrow{\Delta}_{v}(k_1)$ & $\overrightarrow{\Delta}_{v}(k_2)$ & $\overleftarrow{\Delta}_{v}(k_2)$ & $\hat{p}^{(1)}$ & $\hat{p}^{(2)}$ & $\check{p}^{(1)}$ & $\check{p}^{(2)}$ & $d$\\ [0.5ex]
  \hline
  0 & 27 & $1$ & $n.a.$ & $n.a.$ & $\frac{2}{11}$ & $n.a.$ & $1$ & $n.a.$ & $\frac{2}{11}$ & $\frac{81}{22}$ \\ [0.5ex]
  
  1 & 26 & $\frac{1}{4}$ & $\frac{1}{4}$ & $1$ & $0$ & $1$ & $\frac{1}{4}$ & $\frac{1}{4}$ & $0$ & $\frac{29}{24}$\\ [0.5ex]
  
  2 & 25 & $\frac{1}{3}$ & $0$ & $\frac{1}{2}$ & $0$ & $\frac{1}{2}$ & $\frac{1}{3}$ & $0$ & $0$ & $\frac{14}{9}$ \\ [0.5ex]
  
  3 & 24 & $\frac{1}{2}$ & $0$ & $\frac{1}{3}$ & $0$ & $\frac{1}{3}$ & $\frac{1}{2}$ & $0$ & $0$ & $\frac{13}{6}$\\ [0.5ex]
  
  4 & 23 & $1$ & $0$ & $\frac{1}{4}$ & $\frac{1}{6}$ &  $\frac{1}{4}$ & $1$ & $0$ & $\frac{1}{6}$ & $\frac{121}{36}$\\ [0.5ex]
  
  5 & 22 & $\frac{2}{11}$ & $\frac{1}{4}$ & $1$ & $0$ & $1$ & $\frac{2}{11}$ & $\frac{2}{11}$ & $0$ & $\frac{89}{66}$ \\ [0.5ex]
  
  6 & 21 & $\frac{1}{5}$ & $0$ & $\frac{1}{2}$ & $0$ & $\frac{1}{2}$ & $\frac{1}{5}$ & $0$ & $0$ & $\frac{6}{5}$\\ [0.5ex]
  
  7 & 20 & $\frac{1}{4}$ & $0$ & $\frac{1}{3}$ & $0$ & $\frac{1}{3}$ & $\frac{1}{4}$ & $0$ & $0$ & $\frac{11}{9}$ \\ [0.5ex]
  
  8 & 19 & $\frac{1}{3}$ & $0$ & $\frac{1}{4}$ & $0$ & $\frac{1}{4}$ & $\frac{1}{3}$ & $0$ & $0$ & $\frac{25}{18}$ \\ [0.5ex]
  
  9 & 18 & $\frac{1}{2}$ & $0$ & $\frac{2}{9}$ & $0$ & $\frac{2}{9}$ & $\frac{1}{2}$ & $0$ & $0$ & $\frac{11}{6}$ \\ [0.5ex]
  
  10 & 17 & $1$ & $0$ & $\frac{1}{5}$ & $\frac{1}{6}$ & $\frac{1}{5}$ & $1$ & $0$ & $\frac{1}{6}$ & $\frac{97}{36}$ \\ [0.5ex]
  
  11 & 16 & $\frac{3}{16}$ & $\frac{1}{6}$ & $1$ & $0$ & $1$ & $\frac{3}{16}$ & $\frac{1}{6}$ & $0$ & $\frac{73}{36}$ \\ [0.5ex]
  
  12 & 15 & $\frac{1}{5}$ & $0$ & $\frac{1}{2}$ & $0$ & $\frac{1}{2}$ & $\frac{1}{5}$ & $0$ & $0$ & $\frac{3}{2}$ \\ [0.5ex]
  
  13 & 14 & $\frac{1}{4}$ & $0$ & $\frac{1}{3}$ & $0$ & $\frac{1}{3}$ & $\frac{1}{4}$ & $0$ & $0$ & $\frac{47}{36 }$ \\ [0.5ex]
  
\end{tabular}
\caption{The first six columns contain the values of the max-forward-slope and min-backward-slope of each allocation and each buyer. Based on these values, the next 5 columns contain the values of the minimum and maximum high and low prices respectively and the minimum discrepancy for the row's allocation.}
\label{table:d2}
\end{table}

\end{proof}



\newpage

\end{document}